\documentclass[preprint]{article}

\usepackage[nonatbib]{neurips_2025}

\usepackage{amsmath}
\usepackage{amssymb}
\usepackage{mathtools}
\usepackage{amsthm}
\theoremstyle{plain}
\newtheorem{theorem}{Theorem}%[section]

\newtheorem{lemma}{Lemma}
\newtheorem{corollary}{Corollary}
\theoremstyle{definition}
\newtheorem{definition}{Definition}

\theoremstyle{remark}

\newtheorem{example}{Example}
\theoremstyle{observation}
\newtheorem{observation}{Observation}
\theoremstyle{claim}
\newtheorem{claim}{Claim}

\usepackage[textsize=tiny]{todonotes}
\usepackage[utf8]{inputenc} % allow utf-8 input
\usepackage[T1]{fontenc}    % use 8-bit T1 fonts
\usepackage{hyperref}       % hyperlinks
\usepackage{url}            % simple URL typesetting
\usepackage{booktabs}       % professional-quality tables
\usepackage{amsfonts}       % blackboard math symbols
\usepackage{nicefrac}       % compact symbols for 1/2, etc.
\usepackage{microtype}      % microtypography
\usepackage{xcolor}         % colors
\usepackage{algorithm}
\usepackage{enumitem} 

\usepackage{algorithm}
\usepackage[noEnd=true]{algpseudocodex}

\usepackage{marginnote}
\usepackage[utf8]{inputenc} % allow utf-8 input
\usepackage[T1]{fontenc}    % use 8-bit T1 fonts
\usepackage{hyperref}       % hyperlinks
\usepackage{url}            % simple URL typesetting
\usepackage{booktabs}       % professional-quality tables
\usepackage{amsfonts}       % blackboard math symbols
\usepackage{nicefrac}       % compact symbols for 1/2, etc.
\usepackage{microtype}      % microtypography
\usepackage{xcolor}         % colors
\usepackage{framed,graphicx,xcolor}

\title{On multiagent online problems with predictions}

\author{%
   Gabriel Istrate \\
   University of Bucharest \\
  Str. Academiei 14, Sector 1, Bucharest, Romania. \\
  \texttt{gabriel.istrate@unibuc.ro} \\
   \AND
   Cosmin Bonchi\c{s} \\
   West University of Timi\c{s}oara \\
   Bd. V. P\^arvan 4 \\
   \texttt{cosmin.bonchis@e-uvt.ro} \\
  \And
   Victor Bogdan \\
   West University of Timi\c{s}oara \\
  Bd. V. P\^arvan 4  \\
   \texttt{victor.bogdan@e-uvt.ro} \\
}

\begin{document}
\definecolor{shadecolor}{gray}{0.9}
%\colorlet{framecolor}{White}
%\renewenvironment{shaded}{%
%  \def\FrameCommand{\setlength{\FrameRule}{0pt}\fboxsep=\FrameSep \fcolorbox{framecolor}{shadecolor}}%
%  \MakeFramed {\advance\hsize-\width \FrameRestore}}%
% {\endMakeFramed}

\maketitle

\begin{abstract}
We study the power of (competitive) algorithms with predictions in a multiagent setting. We introduce a \emph{two predictor framework}, that assumes that agents use one predictor for their future (self) behavior, and one for the behavior of the other players. The main problem we are concerned with is understanding what are the best competitive ratios that can be achieved by employing such predictors, under various assumptions on predictor quality. 

As an illustration of our framework, we introduce and analyze a \emph{multiagent version of the ski-rental problem}. In this problem agents can collaborate by pooling resources to get a \emph{group license} for some asset. If the license price is not met then agents have to rent the asset individually for the day at a unit price. Otherwise the license becomes available forever to everyone at no extra cost. 

In the particular case of perfect other predictions the algorithm that follows the self predictor is optimal but not robust to mispredictions of agent's future behavior;  we give an algorithm with better robustness properties and benchmark it. 
\end{abstract}

\section{Introduction}
\label{intro}
A great deal of progress has been made recently in algorithmics by investigating the idea of augmenting classical algorithms by \emph{predictors} that guide the  behavior of the algorithm \cite{algo-predictions}. Such predictors, often based on Machine Learning methods, can be instrumental in many practical contexts: e.g. in the area of SAT solving portfolio-based algorithms based on predictors \cite{xu2008satzilla} proved so dominant that they eventually led to changes to the rules of the annual SAT competition. Dually, theoretical advances in algorithms with predictions have impacted research directions such as theory of online algorithms \cite{purohit2018improving,banerjee2020improving,gollapudi2019online}, data structures \cite{mitzenmacher2018model,ferragina2020learned,pmlr-v162-lin22f}, streaming algorithms \cite{Jiang2020Learning-Augmented,chen2022triangle}, 
and even, recently, (algorithmic) mechanism design \cite{xu2022mechanism,agrawal2022learning,gkatzelis2022improved,istrate2022mechanism,balkanski2023online,caragiannis2024randomized}.  The great dynamism of this theory is summarized by the \emph{Algorithm with Predictions webpage \cite{algo-predictions-webpage} } which, as of May 2025 contains 250 papers. Some significant ideas include \emph{distributional advice} \cite{pmlr-v139-diakonikolas21a,canonne2025}, \emph{private predictions} \cite{amin2022private}, \emph{redesigning ML algorithms to provide better predictions} \cite{anand2020customizing}, the use of \emph{multiple predictors} \cite{pmlr-v162-anand22a}, of \emph{predictor portfolios} \cite{dinitz2022algorithms}, 
\emph{uncertainty-quantified predictions} \cite{sun2023online}, \emph{costly predictions} \cite{pmlr-v206-drygala23a}, \emph{heterogeneous predictors} \cite{DBLP:conf/icml/MaghakianLH0S023}, \emph{calibrated predictions} \cite{shen2025algorithms}, and even \emph{complexity classes for online problems with prediction} \cite{berg2024complexity}. %\marginnote{Don't forget to add new references.}[-5mm] 

Conspicuously absent from this list is any attempt of \textbf{investigating the issue of supplementing algorithms by predictions \emph{in multi-agent settings}.} A reason for this absence is, perhaps, the fact that the introduction of multiple agents turns the problem from an optimization problem into a game-theoretic problem, and it less clear what the relevant measures of quality of a solution are in such a setting. A dual reason is that \emph{predictions may interact (nontrivially) with individual optimization}: each agent may try to predict every other agent, these predictions might be incorporated into agent behaviors, agents might attempt to predict the other agents' predictions, and so on. A third and final reason is that we simply don't have crisp, natural examples of problems amenable to an analysis in the framework of algorithm with predictions with multiple agents. 

In this paper we take steps towards the competitive analysis of online problems with predictions in multiagent settings. 
We will work in the framework of \emph{online games} \cite{engelberg2016equilibria}. In this setting agents need to make decisions in an online manner while interacting strategically with other agents. They have deep uncertainty about their own behavior and, consequently, use \emph{competitive analysis} \cite{borodin2005online} to define their objective function. 
We further extend this baseline setting by augmenting agent capabilities with \emph{predictors} \cite{purohit2018improving}. %\footnote{To keep things tractable, in this paper we will assume that \emph{agent predictions are only directed towards their own (unknown) behavior}. Specifically, we will assume that the agents have (some degree of partial) information about the strategies of other agents that they act upon but \textbf{do not try/need to predict these agents any further}. 
%This allows avoiding, at least in this first study, the problems associated with %the often paradoxical interaction between prediction and action exemplified by %\emph{Newcomb's problem \cite{campbell1985paradoxes}}. 
%On the other hand agents' own behavior, reflected in the number of days they will participate in the game, is assumed to be unknown to them and subject to prediction.}. 
Agents may choose the degree of confidence in following the suggestions of the predictor. This usually results \cite{purohit2018improving} in a risk-reward type tradeoff between two measures of algorithm performance, \emph{consistency} (the performance of the algorithm when the predictor is perfect) and \emph{robustness} (the performance when the predictor is bad).

\textbf{Our Contributions.} The main \emph{conceptual contributions} of this paper are: 
\begin{itemize}[leftmargin=*]
\item a \textbf{two-predictor framework for competitive online games with predictions:} given such a game, we will assume that agents are endowed with two predictors - one for its future (self) behavior, and one for the (aggregate) behavior of other agents. 
The two limit cases we consider, for each predictor, are (a). the predictor is perfect (b). it is pessimal. 
\item as a testcase for the feasibility of our framework, \textbf{we  define and analyze (in the two-predictor framework) a multiagent version of the well-known ski-rental problem} \cite{borodin2005online}. 
\end{itemize} 
We aim to answer the following questions about this problem, from a single agent perspective:  
\begin{shaded} 
\begin{itemize}[leftmargin=*]
\item[-] what is the best possible competitive ratio of a predictionless algorithm, and how much can one (or both) predictor(s) help improve it? 
\item[-] how do algorithms achieving such optimal improvements look like, and how robust are they to mispredictions? 
\end{itemize} 
\end{shaded} 
The answers we provide are summarized in Table~\ref{default}, and can intuitively be described as follows: 
%In most of our results we will assume that agents' predictors have similar capabilities. 
\begin{shaded} 
\begin{itemize}[leftmargin=*] 
\item[-] the optimal algorithm with no/pessimal self- and others' predictions is $(B+1)$-competitive; it rents on day 1 and pledges the full amount $B$ on day 2. 
\item[-] \textbf{when the others' predictor is pessimal, self-predictors don't help:} no algorithm with self-predictions (even perfect ones) can be better than $B+1$-competitive (although perfect self-predictors help more algorithms reach this optimal competitive ratio). 
\item[-] when the others' predictor is perfect but the self-predictor is pessimal, \emph{the optimal competitive ratio depends on the actions of other players} (Theorem~\ref{thm-full-free}). We characterize all the optimal algorithms. 
\item[-] when others' predictions are perfect, the algorithm that blindly follows the self-predictor is optimal  (1-competitive) when the self-predictor is perfect, but is \emph{not robust to self-mispredictions} (Thm.~\ref{thm:blindly}). 
\item[-] we also give \textbf{an algorithm, parameterized by $\lambda\in [0,1]$, that interpolates between the optimal algorithm with no self-predictions and the one that blindly follows the predictor, 
trading off consistency for robustness, in the style of results from \cite{purohit2018improving} for the one-agent fixed-cost ski-rental with predictions.} The robustness of the algorithm depends directly on  $\lambda$. As for the consistency, we give (Theorem~\ref{alg-predictions}) an analytic guarantee valid for an arbitrary input. We benchmark the robustness our algorithm through experiments. Our experiments suggest that \textit{the algorithm we propose is promising in scenarios where the others' predictor is imperfect} as well. 
\end{itemize} 
\end{shaded}

%We then investigate (Sections~\ref{sec:5} and~\ref{sec:6}) the impact of using such an algorithm in competitive ratio equilibria, as well as the new, approximate equilibria with predictions that arise this way. 

\textbf{Broader Impact.} \emph{Predictions} have many potential strategic implications in multiagent systems. For instance agents might exploit their predictions of other agent behaviors in adversarial settings. At the same time, mispredictions might disrupt critical collaborative decisions, requiring robust safeguards. 

Though our paper is conceptual, we hope that ideas and techniques similar to those developed in this paper can inspire research that ultimately applies to more practical frameworks, such as systems of interacting Markov decision processes. At the same time, improved algorithms that intelligently exploit predictions could help in a variety of settings, e.g. optimizing  resource sharing in distributed systems (an attractive example is that of using Traffic/Navigator Apps such as Waze or Sygic).  

\textbf{Outline.} The plan of the paper is as follows: in Section~\ref{sec:prelim} we review some notions we will use and define the main problem we are interested in. The Two Predictor Model is investigated next: the case of pessimal self-predictors in Section~\ref{two:pessimal}, and that of optimal self-predictors in Section~\ref{sec:prophet}. We give (Section~\ref{sec:4}) an algorithm for such an agent displaying a tradeoff between robustness and consistency. We benchmark our algorithm (Section~\ref{sec:experiments}) and discuss further issues and open problems (Section~\ref{sec:conclusions}). \textbf{We defer all proofs to the Supplemental Material}.

\begin{table}[htp]
\caption{Summary of optimal competitive ratios (and algorithms realizing them).}
\begin{center}
\begin{tabular}{cccc}
\toprule
Self-predictor & Others Predictor & Competitive Ratio & Optimal Algorithm(s)  \\
\midrule
none/pessimal & none/pessimal & $B+1$ & described in Thm.~\ref{pessimal-no}  \\
%\hline
optimal & none/pessimal & $B+1$ & described in Thm.~\ref{opt:self:noother}  \\
%\hline
none/pessimal & optimal & $\stackrel{\mbox{depends on other players' behavior}}{\mbox{(explicit formula in Thm.~\ref{thm-full-free})}}$ & described in Thm.~\ref{thm-full-free}\\
%\hline
optimal & optimal & $1$ & Algorithm~\ref{zero-alg} \\
\bottomrule
\end{tabular}
\end{center}
\label{default}
\end{table}%

\section{Preliminaries}
\label{sec:prelim} 

Given a sequence $a=(a_n)_{n\geq 1}$ we will denote, abusing notation, by $argmin(a)$ both the set of all indices $i\geq 1$ such that $a_i=min\{a_{n}:n\geq 1\}$ and the \emph{smallest index} of an element in this set. Also define $largmin(a)$ to be the \emph{largest index} of an element in $argmin(a)$. 

We assume knowledge of standard concepts in the competitive analysis of online algorithms, at the level described e.g. in \cite{borodin2005online}. We will consider online algorithm $A$ dealing with inputs $x$ from a set $\Omega$ of all legal inputs. We say that $A$ is \emph{$\gamma$-competitive} iff there exists a constant $d$ (that may depend on $\gamma$, but not on $x$) such that $cost(A(x))\leq \gamma \cdot cost(OPT(x))+d$ for all $x$. Additionally, $A$ is \textit{strongly $\gamma$-competitive} when $d=0$. By forcing somewhat the language, we refer to the quantity 
$\frac{cost(A(x))}{cost(OPT(x))}$ as the \emph{competitive ratio of algorithm $A$ on input $x$}, denoted by $c_{A}(p)$. We will also denote by $c_{OPT}(x)$ the optimal competitive ratio of an algorithm on input $x$. This assumes the existence of a class $\mathcal{A}$ of online algorithms from which $A$ is drawn, hence $c_{OPT}(x)=\inf(c_{A}(x):A\in \mathcal{A}).$

An online algorithm \emph{with predictions} bases its behavior not only on an input $x\in \Omega$ but also on a \emph{predictor $\widehat{Z}$} for some unknown information $Z$ from some "prediction space" $\mathcal{P}$. %We assume that there exists a metric d on $\mathcal{P}$ that allows us to measure (a posteriori) the quality of the prediction by the prediction error $\eta = d(X,\widehat{X})$. 
We benchmark the performance of an 
an online algorithm $A$ with predictions on an instance $x$ as follows: instead of a single-argument function $A(x)$ we use a \emph{three-argument function} $A(x,Z,\widehat{Z})$. The following idea was proposed in \cite{purohit2018improving}, for measuring the performance of online algorithms with (self) predictions: given a predictor $\widehat{Z}$ for the future, the performance of an online algorithm is measured not by one, but by a pair of competitive ratios: \textit{consistency},  the competitive ratio achieved in the case where the predictor is perfect, and \textit{robustness}, the competitive ratio in the case where the predictor is pessimal.  We stress that in \cite{purohit2018improving,gollapudi2019online,angelopoulos2020online,wei2020optimal} consistency/robustness are \textbf{not} aggregated into a single performance measure. Instead, one computes the Pareto frontier of the two measures, and/or optimizes consistency for a desired degree of robustness.

\begin{definition} An online algorithm with predictions $A$ is \emph{$\gamma$-robust  on input x} iff there exists a constant $d$ (that may depend on $\gamma$, but not on $x,Z,\widehat{Z}$) such that $cost(A(x,Z,\widehat{Z}))\leq \gamma \cdot cost(OPT(x))+d$ for all values $Z,\widehat{Z}$\footnote{Strong $\gamma$-robustness is obtained when $d=0$. Note that the result in \cite{purohit2018improving} 
guarantees strong robustness, whereas our (similar) Theorem~\ref{alg-predictions}  only guarantees the weaker version of robustness. This, however, seems unavoidable for technical reasons, explained/discussed in Section~\ref{def:expl} of the Appendix.}, and \emph{$\beta$-consistent} if $cost(A(x,Z,\widehat{Z}=Z))\leq \beta\cdot cost(OPT(x))$ for all $Z$. If this holds for every $x$ (with the same $d$) we will simply refer to the algorithm as $\gamma$-robust and $\beta$-consistent. 
\label{def:one}
\end{definition} 

 %For instance Algorithm 2 from \cite{purohit2018improving}, parameterized by a real number $\lambda\in (0,1]$, would be called $(1+\frac{1}{\lambda})$-robust in that paper (which assumes that $c_{OPT}(p)=2$) but is  $1/\lambda$-robust under Definition~\ref{def:one}. 

\subsection{Competitive online games with prediction and multiagent ski-rental}

Consider a $n$-person competitive online game \cite{engelberg2016equilibria}. 
A \emph{strategy profile} in such a game is a vector of online algorithms $W=(W_1,W_2,\ldots, W_n)$, one for each player. Given strategy profile $W,$ denote by $BR_i(W)$ the set of strategy profiles where the $i$'th agent plays a best-response (i.e. an algorithm minimizing its competitive ratio) to the other agents playing $W$. Abusing notation, we will also write $BR_i(W)$ for an arbitrary strategy profile in this set. 
A \textbf{competitive ratio equilibrium} is a strategy profile $(W_1,W_2,\ldots, W_n)$ such that each $W_i$ is a best-response to the programs of the other agents.  

A \textbf{competitive online game with predictions} is the variant of competitive online games obtained by assuming instead that agents employ algorithms with predictions. 
The following could be a formalization of this concept that is sufficient for our purposes (we made no attempt to consider the most general setting, and will not actually use it, as we don't prove general results): Assume that each agent $i$ can perform, at every integer time $k\geq 1$, an action from a finite set $\Gamma_i$. We require that every $\Gamma_i$ contains a special action STOP. Player $i$ playing STOP is interpreted as no longer participating in the game. We require that this is a final decision (i.e. once a player stops playing it never returns). Finally, assume that there exists a cost function $c_i: \Gamma_{i}^{\mathbb{N}}\times \mathbb{N} \rightarrow \mathbb{R}_{+}\cup \{\infty\}$, the cost incurred by player $i$ if players play action profile $\gamma$ at moment $k$. We assume that if player $i$ plays STOP in profile $\gamma$ at time $k$ then $c_i(\gamma,k)=0$. 
The \emph{total cost incurred by player $i$ on action profile $\gamma$} is defined as $c_{i}(\gamma):=\sum_{k}c_{i}(\gamma,k)$.

We will consider such games under a \textbf{two-predictor model}: we  assume that every agent $i$ comes with two predictors - a \emph{self predictor} for its own future, and an \emph{others' predictor},  for the actions of other agents.\footnote{Multiagent games where agents attempt to predict the aggregate behavior of other players have been studied before: an example is the \emph{minority game} in Econophysics \cite{minority-games}.} A \emph{self-predictor} is a monotonically decreasing function $s_i:\mathbb{N}\rightarrow \{0,1\}$. $s_i(k)=0$ is interpreted as the player having stopped playing the game by time $k$. Let $\Gamma_{-i}=\prod_{j=1, j\neq i}^{n} \Gamma_j$.  An \emph{others predictor} is a function $o_{i}:\mathbb{N}\rightarrow \mathcal{P}_{f}(\Gamma_{-i})$. $o_{i}(k)$ denotes the predicted finite set of possible action profiles of other players at moment $k$. We assume that $o_i$ and $s_i$ are such that if a player is predicted to play STOP at time $k$ then it will be predicted to play STOP at all subsequent times. 
\begin{observation} 
Since we are using \emph{two} predictors, the possible combinations of pessimal/perfect predictors requires  \textbf{four} competitive ratios to quantify the performance of an algorithm on a given input $p$. 
Optimizing each of these four competitive ratios individually over the set of all deterministic online algorithms (and, possibly,  inputs $x$) yields the values in Table~\ref{default} (the \emph{algorithms} realizing these minimal values will generally vary with the scenario - no single algorithm realizes them all). 
\end{observation} 

\begin{observation} 
We can use the four competitive ratio characterization of quality of an algorithm $A$ on a given input to \textbf{define equilibria with predictions.} Specifically, instead of assuming (as done in competitive ratio equilibria) that agents use a best-response strategy (that may generally be optimally consistent but not robust), we will specify desired degrees of robustness/consistency for strategies employed by each player. \textbf{We defer, however, a formal treatment of equilibria with predictions to subsequent work,} and instead focus here on algorithms with good competitive ratios. 
\end{observation} 

The main problem we are concerned with in this paper is the following \emph{multiagent ski-rental problem}: 

\begin{definition}[Multiagent Ski Rental] 
$n$ agents are initially active but may become inactive (once an agent becomes inactive it will be inactive forever.)  Active agents need a resource for their daily activity. Each day,  active agents have the option to (individually) rent the resource, at a cost of 1\$/day. They can also cooperate in order to to buy a group license that will cost $B>1$ dollars. For this, each agent $i$ may pledge some 
amount $w_i>0$ or \emph{refrain from pledging} (equivalently, pledge 0. We assume that both pledges $w_i$ and the price $B$ are integers.) If the total sum pledged is at least $B$ then the group license is acquired\footnote{Consistent with the agents not knowing the total pledged sum but only the success/failure to acquire a group license, we will assume that when the license is overpledged agents will pay their pledged sums.}, and the use of the resource becomes free for all remaining active agents from that moment on. Otherwise (we call such pledges \emph{inconsequential}) the pledges are nullified. Instead, every active agent must (individually) rent 
the resource for the day.  Agents are strategic, in that they care about their overall costs. They are faced, on the other hand, with deep uncertainty concerning the number of days they will be active. So they choose instead to minimize their competitive ratio. 
\end{definition} 

Algorithms we consider for multiagent ski-rental will be deterministic and nonadaptive (that is, their decision doesn't depend on the precise amounts pledged by other agents in the past. We can easily implement this by assuming that agents are made aware, at every moment, only of the success/failure of their pledging, and \textbf{not} of the total amount pledged.)
%In other words, the total amount pledged by the other agent might be subject to prediction, but is not an input to agent decisions.
 Several classes are relevant: 
\begin{itemize}[leftmargin=*] 
\item[-] Formally, an \emph{algorithm with (self and others) predictions} is a function $f(k,\widehat{T},\widehat{p})$. Here $k$ is an integer (representing a day), $\widehat{T}\geq 1$ is  a prediction of the true active time of the agent, and $\widehat{p}:\mathbb{N}\rightarrow \{0,\ldots, B\}$ is a prediction of the total amounts pledged by the other agents. 

\item[-] An algorithm with predictions is \emph{simple} if whenever $\widehat{T_1},\widehat{T_2}<M_{*}(\widehat{p})$ or $\widehat{T_1},\widehat{T_2}\geq M_{*}(\widehat{p})$ then $f(k,\widehat{T_1},\widehat{p})=f(k,\widehat{T_2},\widehat{p})$. In other words all that matters for the prediction is whether $\widehat{T}\geq M_{*}(\widehat{p})$. \item[-] An algorithm with predictions is \emph{rational} if for all $k,\widehat{T},\widehat{p}$, $f(k,\widehat{T},\widehat{p})\in \{0,B-\widehat{p}(k)\}$. That is, the agent either pledges just enough to buy the license (according to the others' predictor) or refrains from pledging. 
%\item[-] Finally, we will consider \emph{deterministic algorithms with others predictions} (only), $f(k,\widehat{p})$ with two arguments: a day $k$ and a predictor for others' behavior $\widehat{p}:\mathbb{N}\rightarrow \{0,\ldots, B\}$.
%$f(k,\widehat{p})$ is the amount pledged by algorithm on day $k$ if the set of predictions is $\widehat{p}$. 
\end{itemize} 

Finally, we will need the following single-agent generalization of classical ski-rental: 

\begin{definition}[Ski Rental With Varying Prices] Given integer $B\geq 2$, a single agent is facing a ski-rental problem where the buying price varies from day to day (as opposed to the cost of renting, which will always be 1\$). Denote by $0\leq p_i\leq B$ the cost of buying skis \textbf{exactly on day $i$}, and by $P_i = i-1+p_i$ the \textbf{total cost if buying exactly on day $i$} (including the cost of renting on the previous $i-1$ days). % We assume that there exists integers $m, M\geq 1$ such that $m\leq p_i\leq M$. 
Finally, let $p=(p_{i})_{i\geq 1}$. We truncate $p$ at the first $i$ such that $p_i=0$ (such an $i$ will be called a \emph{free day}). Call a day $i$ s.t. $p_i=1$ a \emph{bargain day}. Define $
M_{*}(p)=min(P_i:i\geq 1).$

The problem will arise in two flavors: (a). Prices $p_i$ are known in advance to the agent. (b). \textrm{"everything goes":}  prices $p_i$ can fluctuate adversarially in $0,\ldots, B$. 
\label{ski-rental-known-varying} 
\end{definition}

\iffalse
\textcolor{red}{
We will, generally, deal with problems where agents have similar capabilities: that is they are all predictionless, all can }
\fi

\section{Multiagent Ski Rental with Pessimal Others' Predictions} 
\label{two:pessimal}
The analysis in this section concentrates on the case where the others' predictor is pessimal, and \textit{gives optimal values for \textbf{two competitive ratios} (both equal to $B+1$) of the four}. More specifically: we first identify optimal algorithms with pessimal self-predictions as well. Then we show that, when using pessimal others' predictors, correct self-predictions don't help improving the competitive ratio.

\iffalse
\begin{itemize} 
\item[-] Agents \textbf{do not} assume rational behavior from all the other agents and, instead, assume that everything may be possible. 
\item[-] Agents still do not know the total bids of all the other agents. However, such bids are constrained by the assumption of rational behavior from all other agents. 
\end{itemize} 
\fi 

\subsection{Baseline Case: No/Pessimal Self and Others' Predictions} 

Note that the assumption that others' behavior may be generated adversarially effectively decouples the decisions of individual agents, turning the problem (from the perspective of an individual agent) into a single-agent ski-rental problem with varying prices of the "anything goes" flavor, the price $p_t$ the agent faces at moment $t$ being the amount left unpledged by the other agents.  Given self-prediction $\widehat{T}$, a non-adaptive algorithm can only propose a set of price thresholds $\theta_i(\widehat{T})$, buying at cost $\theta_i(\widehat{T})$ if $p_i$ satisfies $p_{i}\leq \theta_{i}(\widehat{T})$, renting otherwise. 

\begin{theorem} 
No deterministic online algorithm for the ski-rental problem with varying prices in the "anything goes" case can be better than $(B+1)$-competitive. The algorithm that rents on day $1$ and buys at any price on day $2$ (i.e. $\theta_1=0,\theta_2=B$) is the only $(B+1)$-competitive algorithm. Consequently the optimal online algorithm for multiagent ski-rental with pessimal others and self-predictions is to pledge 0 (i.e. rent) on day 1 and $B$ on day 2. It is $(B+1)$-competitive. 
\label{pessimal-no}
\end{theorem} 
\subsection{Pessimal Others' Predictions, Perfect Self-Predictions} 

We now deal with the case where the agent still has pessimal other predictions, but it is capable to perfectly predict the number of days $\hat{T}=T$ it will be active. %We will show that the performance of the best algorithm depends on the value of $\widehat{T}$. %For $\widehat{T}\geq B+1$ no algorithm can improve over the optimal algorithm with no self-predictions. For $\widehat{T}\leq B$ we will be able to obtain an improvement. 
%Note that in this case a non-adaptive algorithm proposes a set of price thresholds $\theta_i(\widehat{T})$. Unlike the baseline case, the price thresholds may actually depend on the value of $\widehat{T}$. The algorithm buys at cost $\theta_i(\widehat{T})$ if the actual price $p_i$ satisfies $p_{i}\leq \theta_{i}(\widehat{T})$, and rents otherwise.
%In particular, no algorithm with (self)-predictions can improve upon the competitive ratio of the algorithm without predictions. 
The result also gives the optimal course of action for every \textit{fixed value} of the parameter $\widehat{T}$. For some values of $\widehat{T}$ the proposed algorithms are better than $(B+1)$-competitive (although the best possible competitive ratio is still $B+1$): 

\begin{theorem} The following are true: 
\begin{itemize}[leftmargin=*]
\item[-] Suppose $\widehat{T}\leq B$. Then no deterministic online algorithm for the ski-rental with varying prices in the "anything goes" case can be better than $\widehat{T}$-competitive. Algorithms with thresholds $\theta_1(\widehat{T})=0$, $\theta_{2}(\widehat{T})\leq \widehat{T}-2, \theta_{3}(\widehat{T})\leq \widehat{T}-3, \ldots, \theta_{\widehat{T}}(\widehat{T})= 0$ are the only $\widehat{T}$-competitive ones.  
\item[-] Suppose $\widehat{T}=B+1$. Then the optimal algorithms are the ones with thresholds $\theta_1(\widehat{T})=0$, $\theta_{2}(\widehat{T})\leq B, \theta_{3}(\widehat{T})\leq B-1, \ldots, \theta_{B+1}(\widehat{T})\leq 1$. They are $(B+1)$-competitive. 
\item[-] Suppose $\widehat{T}>B+1$. Then the only optimal algorithm in this case is the one in the case without predictions ($\theta_1(\widehat{T})=0,\theta_2(\widehat{T})=B$, that is $(B+1)$-competitive) 
\end{itemize} 
In particular no algorithm with perfect self-predictions can be better than $(B+1)$-competitive, same competitive ratio as that of algorithms with no/pessimal self-predictions. 
\label{opt:self:noother}  
\end{theorem} 
 
\iffalse
\subsection{Consistency versus Robustness} 

In the previous result, Theorem~\ref{opt:self:noother}, we were able to improve the competitive ratio in the case where the predicted number of days was $\widehat{T}\leq B$. As we show now, this comes at a price: if the actual number of days $T$ varies, then the performance of the optimal algorithms degrades arbitrarily badly: 

\begin{theorem} 
The algorithms in case (a). of Theorem~\ref{opt:self:noother} 
have infinite competitive ratio when the predicted number of days is $\widehat{T}\leq B$ but the actual number of days $T$ is arbitrary.
\end{theorem} 
\begin{proof} 
Consider an arbitrary algorithm with $\theta_1=0$, $\theta_{2}\leq \widehat{T}-2, \theta_{3}\leq \widehat{T}-3, \ldots, \theta_{\widehat{T}}= 0$. Let $T> \widehat{T}$ be arbitrary, and let $p_1=1$, $p_2=\ldots = p_{T-1}=B$, 
\end{proof} 
\fi 

\section{Multiagent Ski Rental with Perfect Others' Predictions} 

\label{sec:prophet}

In this section we study the other extreme in others' prediction: \textbf{agents can perfectly predict the total pledges at every stage of the game}. That is, for every $i$, the others' predictor $\widehat{p}_i$ is perfect. 
Though hardly plausible in a practical context, this scenario can be justified in several distinct ways: 
\begin{itemize}[leftmargin=*] 
\item[-] first, (as explained in the previous section)  this case provides \textbf{the other two competitive ratios} characterizing the performance of an online algorithm with two predictors: the best algorithms with no/pessimal self-predictors are given in Thm.~\ref{thm-full-free}, and those for perfect self-predictors in Thm.~\ref{thm:blindly}. %information is \textbf{not really needed} if all we want is to characterize resulting equilibria in our game: predictions could be interpreted as beliefs that happen to be true. If, on the other hand, we want (as it is the case in practice) that \emph{the agents themselves be able to compute their best response strategies} then an assumption such as the one above seems necessary. 
\item[-] second, the situation such agents face is similar to that of the player in Newcomb's problem \cite{campbell1985paradoxes}, a paradox that has attracted significant interest in Decision Theory and Philosophy: in Newcomb's problem the daemon knows what the player will do and can act on this info. In our case, from the agent's perspective it's \emph{the other agents} that may behave  this way. So one could say that \textbf{what we are studying when investigating best responding with perfect others' predictions is game theory with Newcomb-like players}. %Equilibria for players with such behaviors have been recently investigated \cite{fourny2020perfect}. See Section~\ref{sec:conclusions} for more discussions. 
\iffalse
\item[-]  finally, in the setting of \cite{purohit2018improving} the quality of an algorithm with one predictor is measured by two quantities, robustness and consistency, quantifying it performance under worst-case and perfect predictions, respectively. When using \emph{two} predictors, the combination of worst-case and perfect predictors seems to require \textbf{four} numbers. The analysis in this paper, which concentrates on the case where the others' predictor is perfect, \textit{provides \emph{two} of these four numbers}, corresponding to the case where the others' predictor is perfect. 
\fi  
\end{itemize}

We take the single agent perspective in multi-agent ski rental, and investigate the best response to the other agents' behavior. Since the agent learns (via the others' predictor) what the other agents will pledge (and is only unsure about the number of days it will be active),  it knows \emph{how much money it would need to contribute on any given day to make the group license feasible}. This effectively reduces the problem of computing a best response to the ski-rental problem with known varying prices, as described in Definition~\ref{ski-rental-known-varying}, to which the following notations apply: 

\begin{definition} For a given sequence $p=(p_i)_{i\geq 1}$ of prices, infinite or truncated at a certain index $n$, and $t\geq 1$, let $M_t(p)=min \{P_i: 1\leq i \leq t, P_i \mbox{ is defined}\}$ and $i_t(p)=min\{j: P_j=M_t(p)\}$. We have $1\leq M_{t+1}(p)\leq M_{t}(p)$, hence the sequence $(M_{t}(p))$ has the limit $M_{*}(p)$ as $t\rightarrow \infty$. Let $i_{*}=i_{*}(p)=min\{i: P_{i}=M_{*}(p)\}$. If $p$ represents a sequence of buying prices then $i_{*}(p)$ is the first day where the absolute minimum total cost of ownership, $M_{*}(p)$, is reached, should we commit to buying.  Sequence $i_t(p)$ is increasing and stabilizes to $i_{*}(p)$ as $t\rightarrow \infty$. %\textcolor{blue}{Also define $i^{*}(p)$ to be \textbf{the last day} where the absolute minimum total cost of ownership is reached, should we commit to buying.}
Finally, define $k(p)=min\{t: min(P_1,\ldots, P_{t})\leq t)$, $Q_t=min\{P_{i}:i\geq t\}\mbox{ and } q_{t}=min\{p_{i}:i\geq t\}$.
\end{definition}

\begin{observation} If $t\geq i_{*}(p)$ then it is optimal to buy on day $i_{*}(p)$, or rent (whichever is cheaper). The optimal cost is $OPT_{t}=min(t,M_{*}(p))$. If $t<i_{*}(p)$ then it is optimal either to buy on day $i_t(p)$ or rent (whichever is cheaper).  The optimal cost is $OPT_{t}=min(t,P_{i_t}(p))= min(t,P_1,\ldots P_t)$
\end{observation} 

In fact, we can say more: 

\begin{lemma} 
Suppose there is no free day on day $M_{*}(p)+1$ and no bargain day on day $M_{*}(p)$. Then $k(p)=P_{i_{*}(p)}=M_{*}(p)$, and $OPT_{t}=t\mbox{ if } t<M_{*}(p)$, $OPT_{t}=M_{*}(p)\mbox{ if } t\geq M_{*}(p)$. 
%\[
%OPT_{t}=\left\{\begin{array}{cc}
%t, & \mbox{ if } t<M_{*}(p) \\%
%M_{*}(p), & \mbox{ if } t\geq M_{*}(p)
%\end{array}
%\right.
%\]
\label{lemma-opt}
\end{lemma} 
%\begin{example} 
%For classical, fixed-cost  ski-rental, $P_i=B+i-1$, $k(p)=M_{*}(p)=B$, $i_{*}(p)=1$%, $i_{*}(p)=1$ irrespective of the value of $\widehat{T}$ and $r_1(p)=B$
%. 
%\label{ex3}
%\end{example} 

 The following result describes the optimal deterministic competitive algorithms for this latter problem, hence the best responses for individual agents in our case:   
 
\begin{theorem} Optimal behavior in Ski Rental with Known Varying Prices is as follows: 
\begin{itemize}[leftmargin=*]
\item[(a).] if a free day exists on day $d=M_{*}(p)+1$ or a bargain day exists on day $d=M_{*}(p)$, then the agent that waits for the first such day, renting until then and acting accordingly on day $d$ is 1-competitive. The only other 1-competitive algorithm exists when the first bargain day $b=M_{*}(p)$ is followed by the first free day $f=M_{*}(p)+1$; the algorithm rents on all days, waiting for day $f$.  
\item[(b).] Suppose case (a) does not apply, i.e. the daily prices are $p_1,p_2,\ldots, p_n \ldots $, $1< p_i\leq B$ for $1\leq i\leq M_{*}(p)$, $p_{M_{*}(p)+1}\geq 1$. Then no deterministic algorithm can be better than $c_{OPT}(p)$-competitive, where 
$ c_{OPT}(p):=
min(\{\frac{P_r}{r}:r\leq M_{*}(p)\}\cup\{ \frac{Q_{M_{*}(p)}}{M_{*}(p)}\})$. Algorithms that buy on a day in $argmin(\{\frac{P_r}{OPT_r}:1\leq r\leq M_{*}(p)\}\cup \{\frac{P_r}{M_{*}(p)}: r\geq M_{*}(p), P_r=Q_{M_{*}(p)}\})$ are the only ones to achieve this ratio. 
\end{itemize} 
\label{thm-full-free}
\end{theorem} 
%In other words, by Theorem~\ref{thm-full-free} the best response of an agent in the multiagent ski-rental problem is either to pledge 0 and wait for a free/bargain day, or to pledge on the optimal day given by Theorem~\ref{thm-full-free} (b) (maybe at full cost) and pledge 0 on earlier days.  

\section{Multiagent Ski Rental with Perfect Others' and Self-Predictions} 

We next turn to the setting with perfect others' predictions \textbf{and} self-predictions. We note that blindly following this second predictor is \emph{not} a robust strategy: 

\begin{theorem} Algorithm~\ref{zero-alg} is simple and 1-consistent, but it is $r$-robust for no $r>0$. 
\label{thm:blindly} 
\end{theorem}

\begin{algorithm}[tb]
    %\underline{function Euclid} $(a,b)$\;
    \textbf{Input: $(p,T)$}\\
       \textbf{Predictor: $\widehat{T}\geq 1$, an integer.}\\
       \vspace{-5mm}
       \begin{algorithmic}
    \If{$(\widehat{T}\geq min(P_{i}:i\leq \widehat{T}))$}
      \State{ 
          buy on any day $i_1\leq T$ such that }
       \State{
          $P_{i_1}=min(P_{i}:i\leq \widehat{T})$\;
       }
       \Else
       \State{
        keep renting.\;
      }
      \EndIf
      \end{algorithmic} 
    \caption{An algorithm that blindly follows the predictor.}
    \label{zero-alg}
\end{algorithm}

%We will denote by $M_C(p)$ the price paid by the optimal competitive algorithm, as described in Theorem~\ref{thm-full-free}. In other words, in case (b) of the theorem $M_{C}(p)=P_{j(P)}$. 

\section{Trading off Consistency for Robustness for Perfect Others' (but imperfect Self-) Predictions}
\label{sec:4}

Still assuming perfect others' predictors, we next consider the case when self-predictions are imperfect. It is still the case that \emph{individual agent errors are effectively decoupled}:  for every agent $i$ we can characterize its performance in a strategy profile $(W_1,W_2, \ldots, W_n)$ by two such numbers, $\gamma_i$ and $c_i$, \textbf{which depend on $W_i$ only}, as the agent knows what the other agents will pledge. This reduces the problem to a ski-rental problem with varying prices. We will drop index $i$ and reason from this single-agent perspective. 

We will give an algorithm with predictions for this latter problem, that provides a more robust approach  to multiagent ski-rental. In the fixed-price case the relevant results are those in \cite{purohit2018improving} (improved in \cite{banerjee2020improving}), as well as the lower bounds proved independently in \cite{gollapudi2019online,angelopoulos2020online,wei2020optimal}.  %In general, the fact that one agent makes errors in its own prediction may influence the performance of other agents. %In the baseline model of this section (which we call \emph{the pledging model with full information}) we will assume that \textbf{agents know the behavior of all other agents.} %In the open-source framework this could be implemented by assuming that agents can access not only the codes of other agents' programs, but of their predictors as well. 
  %On the other hand the framework we consider requires another type of error. Define 
%\begin{equation} 
%\zeta=max[M_{\widehat{T}}(p)-M_{T}(p),0]. 
%\end{equation} 
%Informally $\zeta$ is the increase in the best buying price due to under-estimating the time $T$. 

\begin{definition} Given a sequence of prices $p=(p_i)_{i\geq 1}$ for the single agent ski-rental problem with varying prices, let $r_0=r_0(p)$ and $r_1=r_1(p)$ be defined as 
%\marginpar{\textcolor{blue}{Note: $r_0$ is argmin to make the proof of Lemma 4.7 work.}}
$r_{0}=largmin(\{r: P_{r}=M_{*}(p)\})$, $
r_{1}=argmin(\{\frac{P_r}{OPT_r}:1\leq r\leq M_{*}(p)\}\cup \mbox{  } \{\frac{P_r}{M_{*}(p)}: r\geq M_{*}(p), P_{r}=Q_{M_{*}(p)}\}).$
%\textcolor{red}{Q: Do we need to take $r_1$ to be the largest index? Where is this needed? What does the program implement?}
\end{definition} 

\begin{definition} A \emph{simple algorithm with predictions $A$} for the ski-rental problem with varying prices is specified by functions $f_1,f_2$, with $f_1(p)\leq r_1(p)$, and works as follows: 
(1). if $p$ contains a bargain day $i=M_{*}(p)$ or a free day $i=M_{*}(p)+1$, then the algorithm $A$ waits for it (and is 1-competitive) 
(2). if $\widehat{T}\geq M_{*}(p)$ then $A$ buys on day $f_1(p)\leq r_1(p)$. 
(3). else the algorithm buys on day $f_2(p)\in \mathbb{N}\cup \{\infty\}$. 
\label{def-simple-algs} 
\end{definition} 

\begin{theorem} On an input $p$ of type (2) or (3), the simple algorithm with predictions specified by $f_1,f_2$ is (a). $\gamma$-robust for no $\gamma$, if $f_2(p)=\infty$.  
(b). $max(\frac{P_{f_1(p)}}{M_{*}(p)}, \frac{P_{f_2(p)}}{f_2(p)})$-consistent if $f_2(p)< M_{*}(p)$, $\frac{P_{f_1(p)}}{M_{*}(p)}$ otherwise, and $max(\frac{P_{f_1(p)}}{OPT_{f_1(p)}(p)},\frac{P_{f_2(p)}}{OPT_{f_2(p)}(p)})$-robust, when $f_{2}(p)<\infty$. 
\label{thm:upper}
\end{theorem} 

We now apply Theorem~\ref{thm:upper} to create more robust algorithms for multiagent ski-rental,  proposing Algorithm~\ref{zeroprime-alg} below, that attempts to "interpolate" between the behaviors of algorithms in Theorems~\ref{zero-alg} and~\ref{thm-full-free}.  For every $\lambda\in (0,1]$ the algorithm is simple (see Definition~\ref{def-simple-algs}). 
The next result shows that it is $\frac{1}{\lambda}$-robust.  However, \textbf{we cannot give a purely mathematical bound for consistency. Instead we give an input-specific guarantee.} However, (a). we show  (Lemma~\ref{lemma-opt2}) that for every $p$ this guarantee improves as $\lambda$ decreases (b). in Section~\ref{sec:experiments} we  experimentally benchmark the average consistency of Algorithm~\ref{zeroprime-alg}, notably its dependency on the prediction error.   
%Technically it is \emph{not} a simple algorithm with predictions (in the sense of Definition~\ref{def-simple-algs}) but can be viewed as "taking the best" of two simple algorithms with predictions: 

\begin{algorithm}
   
    %\underline{function Euclid} $(a,b)$\;
    \textbf{Input: $(p,T)$}\\
       \textbf{Predictor: $\widehat{T}\geq 1$, an integer.}\\
       \vspace{-5mm}
       \begin{algorithmic}
    \If {we are in Case (a) of Theorem 2}
    \State {act accordingly}
    \Else
    %\STATE{let $i_{*}(p) =argmin(P_r:1\leq r\leq \infty)$}
    %\STATE{let $r_{1}=r_1(p)=argmin(\frac{P_r}{OPT_r}:1\leq r\leq M_{*}(p))\mbox{ or } r\geq M_{*}(p), P_{r}=Q_{M_{*}(p)})$}\\
    \State {choose $r_2=r_{2}(p,\lambda)$ to be the first day $\geq \lceil (1-\lambda) (r_0-1)+ \lambda r_1\rceil$ such that $P_{r_2}- \lambda OPT_{r_2}=min(P_{t}-\lambda OPT_{t}:t\geq \lceil (1-\lambda) (r_0-1)+ \lambda r_1\rceil ) $;\\
     if $P_{r_2}>P_{r_1}$ let $r_2=r_1$}
%    \STATE{\textcolor{red}{ if $P_{r_2}>P_{r_1}$ let $r_2=r_1$}}
%\STATE{ }
    \State {choose $r_{3}=r_{3}(\lambda,p)$ so that $\frac{P_{r_{3}}}{OPT_{r_{3}}(p)}\leq \lambda-1+\frac{1}{\lambda}c_{OPT}(p)$  minimizing the ratio $\frac{P_{r_{3}}}{r_{3}}$}
        \If{$(\widehat{T}\geq M_{\widehat{T}}(p))\;$}
      \State{ 
          buy on day $r_2$ (possibly not buying at all if $r_2 >T$)}
      % \STATE{}
       \Else
       \State{
        buy on day $r_3$ (possibly not buying at all if $r_3>T$)}
       % \STATE{ }
      \EndIf
      \EndIf
      \end{algorithmic} 
    \caption{An algorithm displaying a robustness-consistency tradeoff.}
    \label{zeroprime-alg}
\end{algorithm}

\begin{theorem} The competitive ratio of  algorithm~\ref{zeroprime-alg} on input $(p,T)$, prediction $\widehat{T}$,  is 
%\begin{equation}
$c_{A}(p,T;\widehat{T})\leq \lambda - 1+\frac{1}{\lambda}c_{OPT}(p).$
%\end{equation} 
%Consequently, Algorithm~\ref{zeroprime-alg} is 
%$(1+\frac{1}{\lambda}(c_{OPT}-1))$-robust. 
On the other hand, if $T=\widehat{T}$ (i.e. $\eta=0$) then %the algorithm has competitive ratio (a.k.a. consistency) bounded by:  
%\begin{equation}
$ c_{A}(p,T;T)\leq max(\frac{P_{r_2}}{M_{*}(p)}, \frac{P_{r_{3}}}{r_{3}})\mbox{ if }r_3\leq M_{*}(p)$, 
%\end{equation} 
%\begin{equation}
$ c_{A}(p,T;T)\leq \frac{P_{r_2}}{M_{*}(p)},\mbox{ otherwise. }$
%\end{equation}
\label{alg-predictions} 
\end{theorem} 

\begin{example} 
\label{fixed-price}
In the case when $P_{t}=B+t-1$ for all $t$ (i.e. classical ski-rental with fixed price $B$, which has a single input $p$) we have $i_{*}=1$, $P_{i_{*}}=r_1=M_{*}(p)=B$ and $P_{t}-t=B-1$ for all $t\geq 1$. From $i_{*}=r_0=1$ and Lemma~\ref{lemma-opt}, it follows that $r_2=\lceil \lambda B\rceil$. $\frac{P_{r_{3}}}{M_{*}(p)}=1+\frac{r_{3}-1}{B}$, so to make (for $r_{3}\geq B$) $\frac{P_{r_{3}}}{M_{*}(p)}\leq \lambda-1+\frac{1}{\lambda}c_{OPT}(p)=\lambda-1+\frac{1}{\lambda}\frac{2B-1}{B}$ we need $r_{3}+B-1\leq B(\lambda-1)+\frac{2B-1}{\lambda}$. 
The largest value (which minimizes the ratio $\frac{P_{r_{3}}}{r_{3}}$) is 
$r_{3}=\lfloor B\lambda + (2B-1)(\frac{1}{\lambda}-1)\rfloor$. In this case our Algorithm~\ref{zeroprime-alg} is a small variation on Algorithm 1 in \cite{purohit2018improving}:  since we proved that $r_2=\lceil \lambda B\rceil$) we can see that our algorithm performs identically to Algorithm 1 of \cite{purohit2018improving} in the case $\widehat{T}\geq M_{\widehat{T}}(P)$, but it may buy slightly earlier/later in the opposite case. On the other hand \textbf{our algorithm has the same consistency/robustness guarantees as those available for this latter algorithm}. Indeed, the robustness bound was (by design) identical to that of Algorithm 1 of \cite{purohit2018improving}. As for consistency, we show that $B\lambda + (2B-1)(\frac{1}{\lambda}-1)> B$. Indeed this is equivalent to $(2B-1)\frac{1-\lambda}{\lambda}> B(1-\lambda)$, or $\frac{1}{\lambda}> \frac{B}{2B-1}$, which is trivially true, since $B>1$ and $\lambda\leq 1$. So $r_3>   B=M_{*}(p)$. Since $\frac{P_{r_2}}{M_{*}(p)}=\frac{B+\lceil \lambda B\rceil-1}{B}\leq \frac{B(1+\lambda)}{B}=1+\lambda$, the consistency guarantee is, therefore 
$c_{A}(p,T,T)\leq 1+\lambda$, 
the same as the one from \cite{purohit2018improving} for their Algorithm 1. 

%We claim that $r_3\leq \lceil \frac{B}{\lambda}\rceil$.  Since $\frac{P_t+1}{t}=1+\frac{B}{t}$, and $\frac{P_{r_2}+1}{M_{*}(p)}=1+\frac{\lceil \lambda B\rceil}{B}$, to prove this we would have to show that $1+\frac{B}{\lceil \frac{B}{\lambda}\rceil}\leq 
%1+\frac{\lceil \lambda B\rceil}{B}$, equivalently $B^2\leq \lceil \frac{B}{\lambda} \rceil \cdot (\lceil \lambda B\rceil)$, which is clear,  since $\lceil \frac{B}{\lambda} \rceil \geq \frac{B}{\lambda}$, 
%$\lceil \lambda B\rceil \geq \lambda B$.
\iffalse
We show that the inequality is \textbf{not} true for $t=\lceil \frac{B}{\lambda} \rceil - 1$. 

Indeed, let us assume that $B=\lambda k + c$, $k\in \mathrm{N}$, $0\leq c<\lambda$. 
\begin{itemize}[leftmargin=*] 
\item[-] If $c=0$ then the inequality would be equivalent to $k^2\lambda^2 \leq (k-1)\lceil k\lambda^2\rceil$. But $\lceil k\lambda^2\rceil < k\lambda^2+1$ and 
$(k-1)(k\lambda^2+1)=k^2\lambda^2 - k\lambda^2+k - 1= k^2\lambda^2 + k(1-\lambda) -1$
\item[-] 
\end{itemize} 
\fi 
\end{example} 

Note that for a given sequence $p$, the consistency guarantees given in Theorem~\ref{alg-predictions} are increasing as a function of $\lambda$. That is,  decreasing $\lambda$ makes the algorithm more consistent.  
In particular for every fixed input $p$ Algorithm~\ref{zeroprime-alg} becomes optimal for small enough values of $\lambda$:
\begin{lemma} 
(a). For $0\leq \lambda_1<\lambda_2\leq 1$ 
$P_{r_2(\lambda_1,p)}\leq P_{r_2(\lambda_2, p)} \mbox{ and }
\frac{P_{r_3(\lambda_1,p)}}{r_3(\lambda_1,p)} \leq \frac{P_{r_3(\lambda_2, p)}}{r_3(\lambda_2,p)}. 
$
(b). If $\lambda < \frac{1}{M_{*}(p)-r_{0}+1}$ and $\widehat{T}\geq M_{*}(p)$ then Alg.~\ref{zeroprime-alg} buys on day $r_0$ (i.e. at the minimum possible price). 
\label{lemma-opt2}  
\end{lemma}

\section{Experimental evaluation of Algorithm~\ref{zeroprime-alg}} 
\label{sec:experiments}
\label{sec:6}
\begin{figure*}[ht]
%\vskip 0.2in
\begin{center}
\centerline{
\begin{minipage}{.33\textwidth}
  \includegraphics[width=.9\textwidth,height=3cm]
  {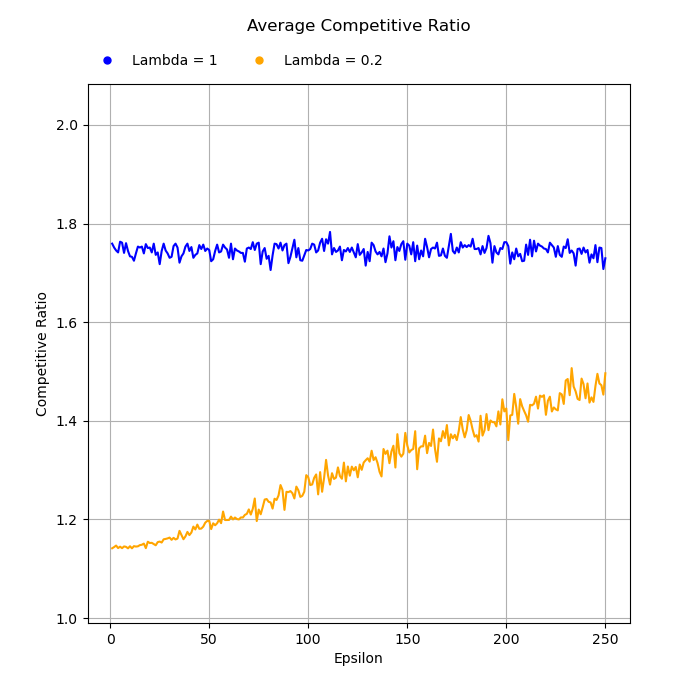}
  \end{minipage} 
  \begin{minipage}{.33\textwidth}
  \includegraphics[width=.9\textwidth,height=3cm]
  {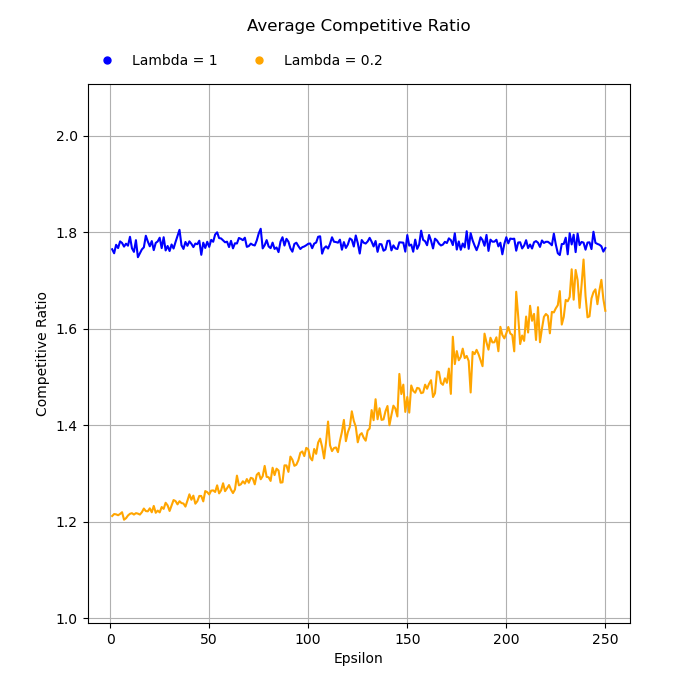}
  \end{minipage}
  \begin{minipage}{.33\textwidth}
  \includegraphics[width=.9\textwidth,height=3cm]
  {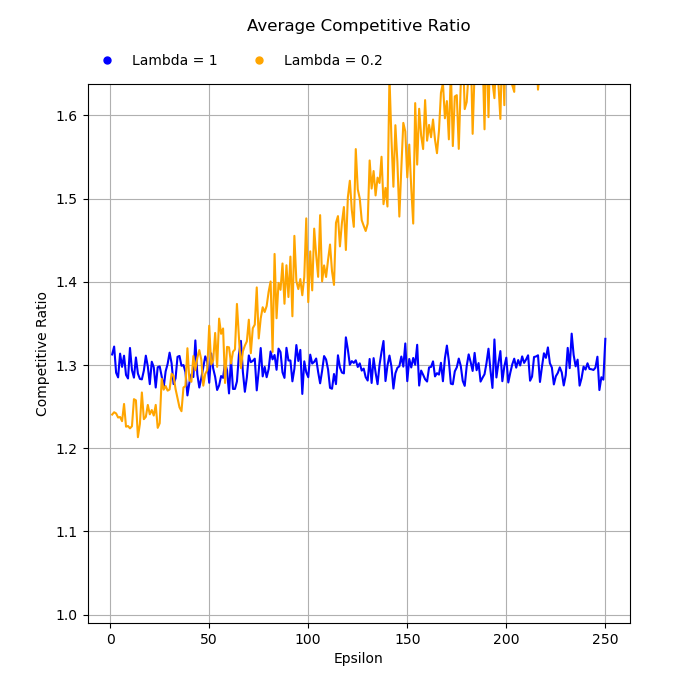}
  \end{minipage}
  }
  \caption{(a,b,c). Average competitive ratio for ski-rental with varying prices under random noise. (a). z=0 (b). z=0.5 (c). z=1. In each case 1000 samples were used for every value of the standard deviation. We used two values of  parameter $\lambda$: $\lambda =1$ (blue) and $\lambda=0.2$ (orange).}
\label{fig1}
\end{center}
\vskip -0.2in
\end{figure*}
Next, we experimentally investigate the performance of Algorithm~\ref{zeroprime-alg} in more realistic, intermediate regimes of the error parameter $\mu=|T-\widehat{T}|$, between perfect and pessimal self-prediction.  The setting mimics the one in \cite{purohit2018improving} as much as possible; in particular, we also measure \emph{average} rather than worst-case competitive ratio. We test the performance of Algorithm~\ref{zeroprime-alg}, as an algorithm for ski-rental with varying prices (this shows that \textbf{one can use it with imperfect others' predictors as well (with reasonable results), if others' aggregate behavior is  random}, rather than adversarial.) 

As with the experiments in \cite{purohit2018improving}, we assume that $B=100$, $T$ is randomly chosen in the interval $[1,4B]$, $\widehat{T}=T+\epsilon$, where $\epsilon$ is normally distributed with average 0 and standard deviation $\sigma$. Prices fluctuate randomly in the interval $[B-\lfloor zB\rfloor,B]$ for $0\leq z\leq 1$. We display three values for $z$: $z=0$ (classical ski rental), $z=0.5$, and $z=0.1$. %We aim to understand how far is the theoretical guarantee provided by Theorem~\ref{eq-one-prediction} from the observed behavior. 
The algorithm performs reasonably well under random noise: Figures~\ref{fig1} (a,b,c) shows this in three scenarios. The conclusion is that \textbf{results are qualitatively similar with those from \cite{purohit2018improving} for all values of $z$.} %As  expected perhaps, the average competitive ratio seems to slightly shift upwards when increasing $z$ for both values of $\lambda$. 
The computational requirements are rather trivial: our code is written in Python, and runs in reasonable time on a Macbook Pro. We will provide an external link to our code/experiments in the final version of the paper. 

\section{Limitations of our results} 

The major limitation of our results is \emph{the lack of an intermediate regime of error for the others' predictor.} The difficult issue is how to define a meaningful "error parameter" for a \emph{sequence} of prices $P_{i}$. A related problem is how to "interpolate" between pessimal and perfect others' predictor, and between the algorithm of Theorem~\ref{pessimal-no} and Algorithm~\ref{zeroprime-alg}. 

Another possible limitation is the fact that \emph{we provide a proof-of-concept analysis of a particular problem, rather than general results}. We feel, though, that this is not really a limitation, given that (a). classical competitive analysis of online algorithms is itself largely problem-specific, with few general results, and (b). our approach proposes an entirely new line of research, that required a proof-of-concept analysis. \textbf{Mitigating these issues is, of course, an interesting research problem.} 

Finally, two technical issues concerning (i). the strength of our analysis and (ii). our Algorithm~\ref{zeroprime-alg} are discussed in detail in Section~\ref{def:expl} of the Appendix. They seem unavoidable (for technical resons) and, at the same time, rather benign. We explain why this is the case in Section~\ref{def:expl}. 

\section{Conclusions and Further Directions} 
\label{sec:conclusions}
%\vspace{-5mm}
Our main contribution has been the definition of a framework for the analysis of competitive online games with predictions, and the analysis of a multi-agent version of the  ski-rental problem. In subsequent work \textbf{we build upon these results by defining a notion of equilibria with predictions.}

Our work leaves many issues open.  For instance, questions of interaction between prediction and optimization similar to those in Newcomb's problem have to be clarified: our results implicitly used \emph{causal decision theory} \cite{weirich2008causal}, but this is not the only option \cite{ahmed2021evidential,fourny2020perfect}.  There exists, on the other hand, a huge, multi-discipline literature that deals with learning game-theoretic equilibria through repeated interactions (e.g. \cite{kalai1993rational,young2004strategic,JMLR:v24:22-0131}). Applying such results to the setting with predictions and, generally, allowing predictions to be \emph{adaptive} (i.e. depend on previous interactions) is an open problem.  %This amounts, in the setting of Theorem~7.2, to deciding for which parameters $(\lambda_1,\ldots, \lambda_n, \mu_1,\ldots, \mu_n)$ do such approximate equilibria exist. 
Note also that \cite{engelberg2016equilibria}  considers a model in which programs $P$ are benchmarked \textit{against themselves}, i.e. \emph{everyone uses $P$}. Optimal programs $P$ are examples of \emph{Kantian equilibria} \cite{roemer2019we,istrategame}. Still, playing the same program does \textbf{not} mean that agents know everything about other players: even if $P$ is known, their active time might be mispredicted. How to meaningfully adapt our two-predictor model to such a setting is not completely clear.  Our results deal with \emph{deterministic algorithms}. Extending the results to probabilistic ones is an interesting challenge. 

Finally, 
a promising direction is related to the Minority game: in this game agents observe the net aggregate effect of past decisions (whether the bar was crowded or not in any of the past $M$ days) and base their action (go/don't go) on a constant number of "predictors" (initially drawn randomly from the set of boolean functions from $\{0,1\}^{M}$ to $\{0,1\}$). They rank (and follow) these predictors based on their recent success. Agents "learn" the threshold at which the bar becomes crowded: the average number of bar-goers fluctuates around $N/2$. When $M$ is sufficiently small  the fluctuations in the number of bar-goers are larger than those for the normal distribution, indicating \cite{johnson2003financial} a form of "system learning": a constant fraction of the agents use the same predictor, that predicts the system behavior.  Whether something similar is true for the others' predictor in our two-predictor model is interesting (and open).

%Finally, in our results we have only employed deterministic  online algorithms. Obtaining similar results for \emph{randomized} online algorithms is, we believe, worthwhile.  
%\newpage
%\newpage

%\bibliographystyle{named}

\bibliographystyle{alpha}
\bibliography{/Users/gistrate/Dropbox/texmf/bibtex/bib/bibtheory}

\clearpage
\newpage

\appendix
\onecolumn

\section{Supplemental Material. Proofs of Theoretical Results}

\subsection{Proof of Lemma~\ref{lemma-opt}}
\iffalse
\textbf{Variant 1:} By the definition of $k(p)$, we have $min(P_{1},P_{2},\ldots, P_{k(p)-1})>k(p)-1$ and $min(P_{1},\ldots, P_{k(p)})\leq k(p)$. Since there are no free/bargain days, 
for $i\geq 1$ we have $P_{i}\geq (i-1)+2=i+1$. So for $i\geq k(p)$ we have $P_{i}>k(p)$, in particular $P_{k(p)}>k(p)$. It follows that $min(P_1,\ldots, P_{k(p)-1})=k(p)$. Since $P_{i}>k(p)$ for $i\geq k(p)$, we must have $i_{*}(p)\in \{1,\ldots, k(p)-1\}$ and $M_{*}(P)=min(P_1,\ldots, P_{k(p)-1})=k(p)$. 

The formula for $OPT_{t}$ follows. 

\textbf{Variant 2:}
\fi 

By the definition of $k(p)$, we have $min(P_{1},P_{2},\ldots, P_{k(p)-1})>k(p)-1$ and $M_{*}(p)=min(P_i:i\geq 1)\leq min(P_{1},\ldots, P_{k(p)})\leq k(p)$.

If $M_{*}(p)=k(p)$ then the result follows. So assume that $M_{*}(p)<k(p)$. We will show that this hypothesis leads to a contradiction.

\begin{lemma} 
$i_{*}(p)\leq k(p)$. 
\end{lemma} 
\begin{proof} 
By the hypothesis if $i\leq M_{*}(p)+1$ then $p_i\geq 1$ and if $i\leq M_{*}(p)$ then $p_{i}\geq 2$. Since for $i\geq k(p)+2$ we have $P_i\geq i-1\geq k(p)+1$ (also for $i= k(p)+1$, provided $k(p)+1$ is not a free day) it follows that under this latter hypothesis $M_{*}(p)=min(P_i:i\geq 1)=min(P_i: i\leq k(p))\leq k(p)$. So if day $k(p)+1$ is not a free day then $i_{*}(p)\leq k(p)$. 

Assume now that $k(p)+1$ is a free day, that is $P_{k(p)+1}=k(p)$. Since $P_r>k(p)$ for $r>k(p)+1$, it follows that 

$M_{*}(p)=min(P_i:i\geq 1)=min(P_1,\ldots, P_{k(p)})\leq k(p)$.  So $i_{*}(p)\leq k(p)$ in this case as well. 
\end{proof} 

The conclusion of Lemma~\ref{lemma-opt} is that $min(P_1,\ldots, P_{k(p)})=M_{*}(p)$. 
If we had $i_{*}(p)<k(p)$ then $M_{*}(p)=P_{i_{*}(p)}=min(P_1,\ldots, P_{k(p)-1})> k(p)-1$, so $M_{*}(p)< k(p)<M_{*}(p)+1$, a contradiction.  We infer the fact that $i_{*}(p)=k(p)$, so $k(p)-1\leq P_{k(p)}=M_{*}(p)$, i.e. $k(p)\leq M_{*}(p)+1$. In fact we have $k(p)\leq M_{*}(p)$ except possibly in the case when we have equality in the previous chain of inequalities, implying the fact that $k(p)$ is a free day. 

We now deal with this last remaining case. We have $M_{*}(p)=P_{k(p)}=k(p)-1$. But this would contradict the hypothesis that no free day $d\leq M_{*}(p)+1$ exists (since $d=M_{*}(p)+1=k(p)$ is such a day). 

The second part of Lemma~\ref{lemma-opt} follows immediately: if $T<M_{*}(p)$ then $T<k(p)$ hence it is better to rent then to buy. If $T\geq M_{*}(p)$ then it is optimal to buy on day $i_{*}(p)=M_{*}(p)\leq T$, at price $M_{*}(p)$.

\subsection{Proof of Theorem~\ref{pessimal-no}}
 
 \begin{proof} %First of all, the algorithm described in the theorem is 2-competitive: if $T_0$, the number of true skiing days, is $\leq B$ then the algorithm will rent. 

The thresholds of the algorithm, $\theta_{i}(\widehat{T})$ depend on the value of the predictor and \textbf{not} on the (unknown) number of actual days $T$. On the other hand it is this latter number that can vary. So from our perspective for every value of $\widehat{T}$ we have a non-adaptive algorithm to analyze. We will, therefore, fix $\widehat{T}$ and write $\theta_i$ instead of $\theta_i(\widehat{T})$. 

Any algorithm such that $\theta_i<B$ for all days $i$ is $\infty$-competitive: just consider the sequence of prices $p_i=B$ for all $i$. On this sequence the algorithm will rent forever, when it could simply buy at a cost of $B$ on the first day. 

So we are left with considering all algorithms $A_r$ for which the first day such that  $\theta_{i}=B$ is day $r$.\footnote{strictly speaking there are multiple such algorithm. $A_r$ will refer to one of them, chosen arbitrarily.} 

Any algorithm such that $\theta_{1}>0$ is $\infty$-competitive: it is possible that $p_1=0$, so $OPT=0$, but the algorithm spends $\theta_{1}$.\footnote{here we take into the consideration our assumption: the fact that $p_1=0$ means that all other agents pledge a total amount of at least $B$; however, the agent's positive pledge is lost.}  

Similarly, for $r\geq 2$, any algorithm such that $\theta_{1}=0$, $\theta_{i}<B$ for all $i<r$ and $\theta_{r}=B$ is no better than $r-1+B$-competitive: consider the sequence of prices $p_1=1$, $p_i=\theta_{i}+1$ for $2\leq i<r$, $p_r=0$. For this sequence we have $OPT=1$, but the algorithm spends $r-1+B$. So its competitive ratio is at least $r-1+B$. This means that if $r>2$ then such an algorithm cannot be $(B+1)$-competitive.  

To complete the proof we have to show that the algorithm described in the theorem ($\theta_1=0, \theta_2=B$) is $(B+1)$-competitive. 

Assume $T_0\geq 1$ is the number of skiing days before a free day arrives (this event might not happen at all, in which case $T_0$ is simply the true number of skiing days): 
\begin{enumerate} 
\item If $T_0=1$ then the algorithm will rent for $1$ day, incurring a cost of $1$. 
\item If $T_0\geq 2$ then the algorithm buys on day $2$, at a total cost of $B+1$. 
\end{enumerate} 
On the other hand, given a sequence of prices $p_1,p_2,\ldots, p_{T_0}$ ($p_{T_0+1}=0$)  the optimum is described as follows: 
\begin{enumerate} 
\item If $T_0\leq B$ then it is optimal to rent for $T_0$ days or buy on the day $i\leq T_0$ with a minimal total cost $i-1+p_{i}$, whichever is better.  
\item If $T_0\geq B+1$ then it is never optimal to rent for at least $B+1$ days: one could buy on or before day $B$ at a lower cost, since $p_1\leq B$. Of course, the optimum is 
$min(i-1+p_i:1\leq i\leq B)$. 
\end{enumerate} 
In any case $OPT\geq 1$. So $ALG/OPT\leq (B+1)/1=B+1$. 
\end{proof}

 \subsection{Proof of Theorem~\ref{opt:self:noother}}
 
 \begin{proof} 
In this case the relevant algorithms are specified by a sequence of threshold prices $(\theta_{1},\theta_{2},\ldots, \theta_{\widehat{T}})$. $\theta_{i}\in \{0,\ldots, B\}$. 

By the same reasoning as the one in the proof of Theorem~\ref{pessimal-no}, 
all algorithms with $\theta_{1}>0$ are $\infty$-competitive. So the optimal algorithms must satisfy $\theta_{1}=0$. 

Consider an arbitrary sequence of prices $p_{1},p_{2},\ldots$ with $p_1>0$.  No optimal algorithm buys on day 1 (since $\theta_1=0$). We can modify the sequence into a sequence with a worse competitive ratio by setting $p_1=1$. Thus $OPT=1$. So we will assume that the worst-case sequences for the competitive ratio have $p_1=1$. 

\begin{claim} 
The competitive ratio of an algorithm $A$ with the properties above is 
\begin{align*}
& c(A)=max(\Theta_2,\ldots, \Theta_{\widehat{T}},\widehat{T})\mbox{ if }\theta_2,\ldots, \theta_{\widehat{T}}<B,\\ 
& c(A)=j-1+B\mbox{ if }\theta_{j}=B, j\leq \widehat{T}\mbox{ minimal with this property.}  
\end{align*}
\end{claim} 
\begin{proof} 
To prove the lower bound side of the inequality we need to display sequences of prices that force the algorithm buy at prices $\Theta_2,\ldots, \Theta_{\widehat{T}}$, respectively, or rent for $\widehat{T}$ days. 

This is easily accomplished: 
\begin{itemize} 
\item $p_1=1$, $p_2=\theta_2+1, \ldots, p_{j-1}=\theta_{j-1}+1, p_j=\theta_j$ for buying on day $j\leq \widehat{T}$. 
\item $p_1=1$, $p_2=\theta_2+1, \ldots,  p_{\widehat{T}}=\theta_{\widehat{T}}+1$, for renting. 
\end{itemize}
The upper bound is just as easy: the algorithm either buys on some day $2\leq j\leq \widehat{T}$ (at cost $P_j$, hence competitive ratio $P_j$, since $OPT=1$), or rents for $\widehat{T}$ days. 

For the second part of the claim take a sequence of prices with $p_1=1$, $p_{2}=\theta_{2}+1, \ldots, p_{j-1}=\theta_{j-1}+1, p_{j}=B$. The algorithm buys on day $j$, at cost $j-1+B$.  
\end{proof} 

Given the claim, the theorem easily follows: 
\begin{itemize} 
\item In the case $\widehat{T}\leq B$, $c(A)\geq \widehat{T}$ for algorithms of the first type, while algorithms of the second type have competitive ratio at least $B+1$. So the optimal value is $\widehat{T}$, which is obtained when $\Theta_{2}, \ldots, \Theta_{\widehat{T}}\leq \widehat{T}$. 
\end{itemize} 
\end{proof}
 
\subsection{Proof of Theorem~\ref{thm-full-free}}

\begin{proof}

\begin{itemize}[leftmargin=*] 
\item[\hspace{2mm}  (a).] Suppose that a free day $f$ shows up before any bargain day (and no later than day $M_{*}(p)$). The algorithm that just waits for the free day is 1-competitive: if $T< f$ then it is optimal to rent, and the agent does so. If $T\geq f$ then it is optimal to get the object for free on day $f$, and the algorithm does so. 

If, on the other hand a bargain $b$ day shows up right before a free day (and  no later than day $M_{*}(p)$) then the algorithm that just rents waiting for the bargain day when it buys is 1-competitive: if $T< b$ then it is optimal to rent, and the agent does so. If $T\geq b$ then it is optimal to get the object on day $b$, and the algorithm does so. 

\item[(b).] Essentially the same proof as for the classical ski rental: we need to compare all algorithms $A_r$ that buy on day $r$, together with the algorithm $A_R$ that always rents.

Consider algorithm $A_r$ that buys on day $r$. Let us compute the competitive ratio of $A_r$. Denote by $T$ the number of actual days. 

\begin{lemma} 
Algorithm $A_r$ is $\frac{P_r}{OPT_r}$-competitive. 
\end{lemma} 
\begin{proof} 
For $T=r$ the competitive ratio of algorithm $A_r$ is indeed 
$\frac{P_r}{OPT_r}$. We need to show that this is, indeed, the worst case, by computing the competitive ratio of algorithm $A_r$ for an arbitrary active time $T$. 

\textbf{Case 1: $T\leq M_{*}(p)$. Then, by Lemma~\ref{lemma-opt}, $OPT_T=T$.} 
\begin{itemize}[leftmargin=*]
\item[-] If $r\leq T$ then the competitive ratio is $\frac{P_r}{T}\leq \frac{P_r}{r}=\frac{P_r}{OPT_r}$. 
\item[-] If $T< r\leq M_{*}(p)$ then the competitive ratio is $\frac{T}{T}=1$.  
\item[-] If $r> M_{*}(p)$ then the competitive ratio is $\frac{T}{T}=1$. 
\end{itemize} 

Therefore, for $T\leq M_{*}(p)$ the worst competitive ratio of $A_r$ is $\frac{P_r}{OPT_r}$. 

\textbf{Case 2: $T>M_{*}(p)$. Then, by Lemma~\ref{lemma-opt},  $OPT_T=M_{*}(p)$.} 
\begin{itemize}[leftmargin=*]
\item[-] If $M_{*}(p)\leq r\leq T$ then the competitive ratio is $\frac{P_r}{M_{*}(p)}=\frac{P_r}{OPT_r}$. 
\item[-] If $r\leq M_{*}(p)<T$ then the competitive ratio is $\frac{P_r}{M_{*}(p)}\leq \frac{P_r}{r}=\frac{P_r}{OPT_r}$.  
\item[-] If $r> T>M_{*}(p)$ then the competitive ratio is $\frac{T}{M_{*}(p)}\leq \frac{r}{M_{*}(p)}=\frac{r}{OPT_r}$. 
\end{itemize} 

\end{proof} 

To find the optimal algorithm we have to compare the various ratios $\frac{P_r}{OPT_r}$. Note that 
for $r\geq M_{*}(p)$ only the algorithms $A_r$ such that $P_{r}=Q_{M_{*}(p)}$ are candidates for the optimal algorithm, since the numerator of the fraction $\frac{P_r}{OPT_r}$ is $M_{*}(p)$. 

Therefore an optimal algorithm $A_T$ has competitive ratio 
\[
c_{OPT}(p)=min(\frac{P_r}{r}:1\leq r\leq M_{*}(P), \frac{Q_{M_{*}(p)}}{M_{*}(p)})
\]

We must show the converse, that for every $T$ that realizes the optimum the algorithm $A_T$ is optimal. This is not self evident, since, e.g. on some sequences $p$ day $T$ could be preempted by  a free day $d$, so the algorithm would never reach day $T$ for buying. 

There is one case that must receive special consideration: suppose one of the optimal algorithms involves waiting for the first free day $d>M_{*}(p)+1$, i.e. $Q_{M_{*}(p)}=d-1$. We claim that there is no day $d^{\prime}>d$ such that $P_{d^{\prime}}=Q_{M_{*}(p)}$ (i.e. $d$ is the last optimal day; there can be other optimal days after $M_{*}(p)+1$ but before $d$). The reason is that $P_{d^{\prime}}\geq d^{\prime}-1>d-1=Q_{M_{*}(p)}$. So all optimal strategies are realizable. 
\end{itemize} 
\end{proof} 

\begin{example} 
For an example that shows that the second term is needed, consider a situation where $B=100$, $p_1=\ldots = p_100=100$, but $p_{101}=p_{102}=\ldots = 2$. The optimal algorithm will buy on day 101, i.e. beyond day $M_{*}(p)=100$. 
\end{example}

  \subsection{Proof of Theorem~\ref{thm:blindly}}

\begin{proof} 
The first part is clear: on any input $(P,T=\widehat{T})$ the algorithm will choose the optimal action. 
For the second part, if $\widehat{T}=1$ then the algorithm will rent for $T$ days, while OPT pays at most $P_1\leq M$. So the competitive ratio is at least $\frac{T}{M}$. Making $T\rightarrow \infty$ we get the desired conclusion.
\end{proof}

  \subsection{Proof of Theorem~\ref{thm:upper}}
 Assume first that $p$ is such that $f_2(p)=\infty$. 
 
 Consider a prediction when $\widehat{T}<M_{*}(p)$. e.g. $\widehat{T}=0$. Since $f_2(p)=\infty$, for every $T\geq 1$ 
the algorithm will rent for $T$ steps. For $T\geq M_{*}(p)$ $OPT_{T}=T$. The competitive ratio in this case is $\frac{T}{M_{*}(p)}$, which is unbounded.

 Assume now that $f_2(p)<\infty$. We analyze the various possible cases in the algorithm.  Some of these cases will be compatible with the hypothesis $\eta=0$, i.e. with equality $\widehat{T}=T$. Their competitive ratios will contribute to the final upper bound on consistency. The rest of the bounds only contribute to the final upper bound on robustness. 
 
 \begin{itemize}[leftmargin=*] 
\item{Case 1:} $\widehat{T}\geq M_{\widehat{T}}(P)$, $T\geq M_{*}(P)\geq f_1(p)$ (compatible with $\eta = 0$).  $ALG=P_{f_1(p)}$, $OPT=M_{*}(P)$.
 
\item{Case 2:} $\widehat{T}\geq M_{\widehat{T}}(P)$, $T\geq f_1(p)>M_{*}(P)$ (compatible with $\eta = 0$).  $ALG=P_{f_1(p)}$, $OPT=M_{*}(P)$. 

\item{Case 3:} $\widehat{T}\geq M_{\widehat{T}}(P)$, $f_1(p) >T \geq M_{*}(P)$ (compatible with $\eta = 0$.) $ALG=T\leq P_{f_1(p)}$ (since $T\leq f_1(p)-1\leq P_{f_{1}(p)})$, $OPT=M_{*}(P)$. 

\item{Case 4:} $\widehat{T}< M_{\widehat{T}}(P)$, $T<f_2(p)\leq M_{*}(P)$ (compatible with $\eta = 0$.).  $ALG=T$, $OPT=T$. 
\item{Case 5:} $\widehat{T}< M_{\widehat{T}}(P)$, $f_2(p)\leq T< M_{*}(P)$ (compatible with $\eta = 0$). $ALG=P_{f_2(p)}$, $OPT=T\geq f_2(p)$.
 
\item{Case 6:} $\widehat{T}\geq M_{\widehat{T}}(P)$, $T<f_1(p)\leq M_{*}(P)$. (incompatible with $\eta = 0$.) $ALG=T$, $OPT=T$.
 
\item{Case 7:} $\widehat{T}\geq M_{\widehat{T}}(P)$, $f_1(p)\leq T<M_{*}(P)$ (incompatible with $\eta = 0$).  $ALG=P_{f_1(p)}$, $OPT=T\geq f_1(p)$. 

\item{Case 8:} $\widehat{T}< M_{\widehat{T}}(P)$, $f_2(p), M_{*}(P)\leq T$ (incompatible with $\eta = 0$). $ALG=P_{f_2(p)}$, $OPT=M_{*}(P)$. 

\item{Case 9:} $\widehat{T}< M_{\widehat{T}}(P)$, $f_2(p)>T$, $M_{*}(P)\leq T$ (incompatible with $\eta = 0$). $ALG=T\leq f_{2}(p)-1\leq P_{f_2(p)}$, $OPT=M_{*}(P)$.
\end{itemize} 

The competitive ratios (different from 1) compatible with $\eta=0$ are:   
\begin{equation} 
\frac{P_{f_1(p)}}{M_{*}(P)}, \frac{P_{f_2(p)}}{f_2(p)}
\end{equation} 

The second bound is compatible, however, \textbf{only when $f_2(p)< M_{*}(p)$.}
 Thus the algorithm is $max(\frac{P_{f_1(p)}}{M_{*}(P)}, \frac{P_{f_2(p)}}{f_2(p)})$-consistent when $f_{2}(p)< M_{*}(p)$,  $\frac{P_{f_1(p)}}{M_{*}(P)}$-consistent, otherwise. 
 
 The upper bounds competitive ratios (different from 1) incompatible with $\eta=0$ are:   
\begin{equation} 
\frac{P_{f_1(p)}}{f_1(p)},\frac{P_{f_2(p)}}{M_{*}(p)}
\end{equation}

Thus the algorithm is $max(\frac{P_{f_1(p)}}{M_{*}(P)}, \frac{P_{f_2(p)}}{f_2(p)}, \frac{P_{f_1(p)}}{f_1(p)},\frac{P_{f_2(p)}}{M_{*}(p)})$-robust. So the final bound on robustness is $max(\frac{P_{f_1(p)}}{min(f_1(p),M_{*}(p))},\frac{P_{f_2(p)}}{min(f_2(p),M_{*}(p))})= max(\frac{P_{f_1(p)}}{OPT_{f_1(p)}(p)},\frac{P_{f_2(p)}}{OPT_{f_2(p)}(p)})$

\subsection{Proof of Theorem~\ref{alg-predictions}}

We first revisit Example~1 in the paper, adding to it some extra quantities from the algorithm: 

\begin{example} 
In the classical, fixed-cost case of ski-rental, $P_i=B+i-1$. We have $k(p)=M_{*}(p)=B$, $i_{*}(p)=1$, $i_{*}(p)=1$ irrespective of the value of $\widehat{T}$ and $r_1(p)=B$. 
\label{ex2}
\end{example} 

We start by proving some useful auxiliary results. First, in the conditions of Example~1 in the paper we had $r_{0}(p)\leq r_1(p)\leq M_{*}(p)$. This generalizes, although not completely: 

\begin{lemma} 
If there is no free day $f\leq M_{*}(p)+1$ then we have $r_{0}(p)\leq r_1(p)$.%\leq M_{*}(p)$.
\end{lemma} 
\begin{proof} 

%Inequality $r_1\leq M_{*}(p)$ follows by the definition of $r_1$. 

Clearly $r_0\leq M_{*}(p)$: indeed, for all $r\geq M_{*}(p)+2$ we have $P_{r}\geq M_{*}(p)+1$. Since day $M_{*}(p)+1$ is not free, $P_{M_{*}(p)+1}\geq M_{*}(p)+1$ as well. The claim follows. 

Suppose that $r_1<r_0$. Then, by definition of $r_0$, $P_{r_1}\geq P_{r_0}$. But, by the definition of $r_1$, 
\[
\frac{P_{r_1}}{OPT_{r_1}}\leq \frac{P_{r_0}}{OPT_{r_0}}. 
\]
Since we assumed that $r_1< r_0\leq M_{*}(p)$, this yields 
\[
\frac{P_{r_1}}{r_1}\leq \frac{P_{r_0}}{r_0},
\]
hence 
\[
1\leq \frac{P_{r_1}}{P_{r_0}}\leq \frac{r_1}{r_0}<1,\mbox{ a contradiction.}  
\]
The only remaining possibility is that $r_0\leq r_1$. 
\end{proof} 

\begin{lemma} 
We have $P_{r_1}=Q_{r_1}$ and $P_{r_1}-\lambda OPT_{r_1}=min(P_{i}-\lambda OPT_{i}:i\geq r_1)$.  
\end{lemma} 

\begin{proof} 
By the optimality of $r_1$ for the competitive ratio, for all $i\geq r_1$ we have 
\[
\frac{P_{r_1}}{OPT_{r_1}}\leq \frac{P_i}{OPT_i}. 
\]
So $P_{r_1}/P_{i}\leq OPT_{r_1}/OPT_{i}\leq 1$, which shows that $P_{r_1}\leq P_i$, hence $P_{r_1}=Q_{r_1}$. The inequality $OPT_{r_1}\leq OPT_i$ follows from Lemma~\ref{lemma-opt} and the fact that $r_1\leq i$. 

As for the second inequality, we write 
\[
\frac{P_{r_1}}{OPT_{r_1}}-\lambda \leq \frac{P_i}{OPT_i}- \lambda. 
\]
So $\frac{P_{r_1}- \lambda OPT_{r_1}}{P_{i}-\lambda OPT_{i}}\leq \frac{OPT_{r_1}}{OPT_i}\leq 1$ for $i\geq r_1$. 

In conclusion $P_{r_1}- \lambda OPT_{r_1}\leq P_{i}-  \lambda OPT_{i}$ for all $i\geq r_1$. 
\end{proof} 

\begin{corollary} 
$r_2$ in Algorithm~1 in the main paper satisfies $r_2\leq r_1$. 
\end{corollary}

\begin{lemma} 
$\frac{P_{r_2}}{r_2}\leq \frac{1}{\lambda} c_{OPT}(p)+\lambda - 1$. 
\label{lemma-ub-r2}
\end{lemma} 
\begin{proof} 
By the definition of $r_2$ we have $P_{r_2}- \lambda r_2\leq P_{r_1}- \lambda {r_1}$ so 
\begin{align*}
& \frac{P_{r_2}}{r_2}\leq \lambda+\frac{P_{r_2}-\lambda r_2}{r_2}\leq \lambda+\frac{P_{r_1}-\lambda r_1}{r_2}\leq \lambda+\frac{r_1}{r_2}\frac{P_{r_1}-\lambda r_1}{r_1} \\ & \leq \lambda +\frac{r_1}{\lambda r_1} (\frac{P_{r_1}}{r_1}-\lambda)= \lambda +\frac{1}{\lambda}(\frac{P_{r_1}}{r_1}-\lambda)= \frac{1}{\lambda} \frac{P_{r_1}}{r_1}+\lambda - 1\leq  \frac{1}{\lambda} c_{OPT}(p)+\lambda - 1\\
%\textcolor{red}{\frac{1}{\lambda}c_{OPT}-f(\lambda)(\frac{1}{\lambda}-1)\leq 1+\frac{1}{\lambda}}
\end{align*} 
We have used inequality $r_2\geq (1-\lambda)(r_0-1)+\lambda r_1\geq \lambda r_1$. 

We have also made the assumption that $\frac{P_{r_1}}{r_1}\leq c_{OPT}(p)$. In fact $c_{OPT}(p)=\frac{P_{r_1}}{OPT_{r_1}}$, and $OPT_{r_1}=min(r_1,M_{*}(p))\leq r_1$, 
so this is correct. 
\iffalse
\[
\textcolor{red}{f(\lambda)(\frac{1}{\lambda}-1)\geq \frac{1}{\lambda}(c_{OPT}-1)-1)}
\]

\[
f(\lambda)\geq \frac{\frac{1}{\lambda}\frac{B-1}{B}-1}{\frac{1}{\lambda}-1}=\frac{B-1-\lambda B}{B(1-\lambda)}
\]
\fi 
\end{proof} 

We now return to the proof of the main theorem. 
%By Theorem~4, for each choice of $r_3$ the algorithm will be $max(\frac{P_{r_2}}{M_{*}(P)}, \frac{P_{r_3}}{r_3})$-consistent and $max(\frac{P_{r_2}}{OPT_{r_2}(p)},\frac{P_{r_3}}{OPT_{r_3}(p)})$-robust. %Since $r_2\leq r_1\leq M_{*}(p)$, the robustness guarantee is $max(\frac{P_{r_2}}{r_2},\frac{P_{r_3}}{OPT_{r_3}(p)})$. 

Each upper bound for consistency/robustness has two terms. The ones containing $r_2$ cannot be improved (at least without changing the definition of $r_2$ in Line 6 of Algorithm 2). But we can choose $r_{3}$ to optimize consistency, while still guaranteeing a $1+\frac{1}{\lambda}(c_{OPT}(p)-1)$ upper bound for robustness. Since $\frac{P_{r_2}}{r_2}\leq \lambda-1+\frac{1}{\lambda}c_{OPT}(p)$, this amounts to minimizing $\frac{P_{r_{3}}}{r_{3}}$ subject to $\frac{P_{r_{3}}}{OPT_{r_{3}}(p)}\leq \lambda-1+\frac{1}{\lambda}c_{OPT}(p)$.

The robustness upper bound is thus at most $\lambda-1+\frac{1}{\lambda}c_{OPT}(p)$. The consistency bound is $max(\frac{P_{r_2}}{M_{*}(P)}, \frac{P_{r_{3}}}{r_{3}})$ (or only the first term, whichever is the case). 

\subsection{Why we guarantee robustness instead of strong robustness, as in \cite{purohit2018improving}}
\label{def:expl}

%\fbox{\textcolor{red}{To complete}}

The very careful reader has undoubtedly noted two subtle technical points concerning the definitions/results in the main paper: 
\begin{itemize} 
\item[-] Our definition of robustness is different/weaker than the one employed in \cite{purohit2018improving}. 
\item[-] For fixed-price ski-rental, our Algorithm~\ref{zeroprime-alg} is reminiscent of the algorithm with predictions in \cite{purohit2018improving} and has identical consistency/robustness guarantees, but does \textbf{not} coincide with this algorithm. 
\end{itemize} 

In this section we discuss/explain why these two phenomena take place.

The algorithm from \cite{purohit2018improving} interpolates between the optimally competitive algorithm without predictions for ski rental (that always buys on day $B$), for $\lambda = 1$, and the algorithm for optimal predictions that either buys on day $1$, if $\widehat{T}\geq B$, or rents (i.e. buys on day $\infty$, otherwise). This is simulated, for $\lambda \approx 0$, by buying on a day later than $B$ in the case $\widehat{T}<B$. 

Our algorithm deals with a more complicated problem, that of ski-rental with known varying prices. The optimal day for buying with no (or pessimal predictions) is, as shown by Theorem~\ref{thm-full-free}(b), assuming case (a) of the theorem does \textbf{not} apply, day $r_1(p)$. The smallest possible buying price is obtained, on the other hand, on day $r_0(p)$. Our algorithm tries to interpolate, for $\lambda\in (0,1)$, between buying on days $r_0(p)$, for $\lambda\approx 0$ and $r_1(p)$, for $\lambda =1$.\footnote{The reason one has $r_0(p)-1$ is to allow the buying time to be $r_0(p)$ for $\lambda\approx 0$; this happens because of the ceiling function.} In that it parallels the philosophy of the algorithm from \cite{purohit2018improving}, that buys on day 1 (if $\widehat{T}>B$ and $\lambda\approx 0$) or on day $B$ (if $\lambda=1$). 

In the case of fixed-cost ski-rental the proposed algorithm is simple (according to the terminology of Definition~\ref{def-simple-algs}). The exact functions $f_1(p)$, $f_{2}(p)$ (or, rather, the buying days $f_1,f_2$ for the two cases of the algorithm, since in this scenario there is a single input $p$) are chosen carefully so that the resulting algorithm is $(1+\lambda)$-consistent. They are $f_1=\lceil \lambda B\rceil$, $f_2=\lceil \frac{B}{\lambda}\rceil$. Buying on day $f_1$ when $T=\widehat{T}\geq B$ (which is equal to $M_{*}(p)$, in our notation) ensures that, since the ratio $\frac{f_1-1+B}{B}$ is at most $\frac{B+B\lambda}{B}=1+\lambda$. For $T=\widehat{T}<B$ there are two cases: 
\begin{itemize} 
\item[-] $T< f_2$. In this case both the algorithm and OPT rent, and the competitive ratio is equal to $1$. 
\item[-] $B>T\geq f_2$. In this case the competitive ratio is at most $\frac{B+f_2-1}{T}\leq \frac{B+f_2-1}{f_2}\leq \frac{B+\frac{B}{\lambda}}{\lambda B}=1+\lambda$. 
\end{itemize} 

There are several other requirements we want satisfied, that are easily satisfied for fixed-cost ski rental, but not necessarily so for variable cost ski-rental. 

One such requirement is Lemma~\ref{lemma-opt2}: because of this result the performance guarantee of the algorithm varies predictably with $\lambda$, truly interpolating between the two limit cases. Another requirement is that when $\lambda=1$, $r_{2}(p,1)=r_{3}(p,1)=r_{1}(p)$. That is, when $\lambda = 1$ the algorithm reduces to the one without predictions. 

However, there is a question how to define $r_{2}(\lambda)$, while satisfying these two requirements: in the fixed-case we could simply take it to be $\lceil (1-\lambda) (r_0-1)+ \lambda r_1\rceil$, as the sequence of prices $P_i=i-1+B$ is increasing with $i$. In the variable price case this is, obviously, no longer true. 
\begin{itemize} 
\item[-] By the previous discussion, choosing $r_{2}(p,\lambda)$ to be $\lceil (1-\lambda) (r_0-1)+ \lambda r_1\rceil$ would potentially violate (the first statement of)  Lemma~\ref{lemma-opt2}. 
\item[-] Another possibility (that we have actually tried, in a preliminary version of this paper) was to take 
$r_{2}(p,\lambda)$ to be the  first day $\geq \lceil (1-\lambda) (r_0-1)+ \lambda r_1\rceil$ such that $P_{r_2}=min(P_{t}:t\geq \lceil (1-\lambda) (r_0-1)+ \lambda r_1\rceil ) $. But this definition does not satisfy the second requirement !
\end{itemize} 
It seems to us that \textbf{it is impossible to simultaneously satisfy all these requirements, while extending the algorithm from~\cite{purohit2018improving}, and its analysis, from ski-rental with fixed prices to ski-rental with varying prices.}\footnote{or, at least, we saw no obvious way to do it.} 

The solution we have found is, we believe, a reasonable compromise: 
\begin{itemize} 
\item[-] first, as we prove, it satisfies the two desired requirements.
\item[-] second, while it slightly differs from the algorithm from \cite{purohit2018improving} in case of ski-rental with fixed prices, it leads (as we show in Example~\ref{fixed-price}) to the same guarantees as those from \cite{purohit2018improving}. So the change of algorithms is inconsequential in the fixed-price case. 
\item[-] third, as we show experimentally in Section~\ref{sec:experiments}, the new algorithm behaves similarly to the one in \cite{purohit2018improving}. This is precisely why our first set of experiments attempted to replicate the settings of \cite{purohit2018improving} as close as possible. 
\item[-] Fourth, since $\gamma\geq 1$, our requirement for robustness in Definition~\ref{def:one} is stronger than the one from \cite{purohit2018improving}, implying $\gamma$-robustness in the sense of their definition. The converse is not necessarily true, that's why we say in footnote 6 that our definition offers slightly weaker robustness guarantees: an algorithm that is $\gamma$-robust according to \cite{purohit2018improving} may only be $\gamma^{\prime}$-robust (according to Definition~\ref{def:one}) for some $\gamma^{\prime}>\gamma$. In any case, the change in definition is not an issue for our (analysis of) Algorithm~\ref{zeroprime-alg}.  
\end{itemize} 

\subsection{Proof of Lemma~\ref{lemma-opt2} (a).} 

\begin{proof} 

\textbf{First inequality.} 

First note that 
$M_{*}(p)$ does not depend on $\lambda$, but $P_{r_2}$ may do so. 
For $\lambda_1<\lambda_2$ we have $\lceil (1-\lambda_1) (r_{0}-1)+ \lambda_1 r_1\rceil \leq \lceil (1-\lambda_2) (r_{0}-1)+ \lambda_2 r_1\rceil$, so, by the definition of $r_2$ we have that $r_2(\lambda_2)$ is a candidate in the definition of $r_2(\lambda_1)$, so 
\begin{equation} P_{r_2(\lambda_1)}-\lambda_1OPT_{r_{2}(\lambda_1)}\leq P_{r_2(\lambda_2)}-\lambda_1 OPT_{r_{2}(\lambda_2)}
\label{f1}
\end{equation}
\begin{itemize}[leftmargin=*] 
\item If $OPT_{r_2(\lambda)}=r_{2}(\lambda)$ for $\lambda\in \{\lambda_1,\lambda_2\}$, this translates to 
\begin{equation} 
P_{r_2(\lambda_1)}-\lambda_1 r_2(\lambda_1)\leq P_{r_2(\lambda_2)}-\lambda_1 r_2(\lambda_2). 
\end{equation}
If we could prove that $r_2(\lambda_1)\leq r_2(\lambda_2)$ then  it would follow that $P_{r_2(\lambda_1)}\leq P_{r_2(\lambda_2)}$. 
Assume, therefore, that 
\begin{equation} r_2(\lambda_1)>r_2(\lambda_2).
\end{equation}  
Because of this relation and the definition of function $r_2$ we have 
\begin{equation} 
P_{r_2(\lambda_2)}-\lambda_2 r_{2}(\lambda_2)\leq P_{r_2(\lambda_1)}-\lambda_2 r_{2}(\lambda_1). 
\label{hyp1}
\end{equation} 
But then 
\begin{align*}
& P_{r_2(\lambda_1)}-\lambda_2 r_2(\lambda_1)= 
P_{r_2(\lambda_1)}-\lambda_1 r_2(\lambda_1)+(\lambda_1-\lambda_2)r_{2}(\lambda_1)\stackrel{(10)}{\leq} P_{r_2(\lambda_2)}-\lambda_1 r_2(\lambda_2)+\\ & +(\lambda_1-\lambda_2)r_{2}(\lambda_1)\leq
 P_{r_2(\lambda_2)}-\lambda_2 r_2(\lambda_2)+ (\lambda_2-\lambda_1)r_2(\lambda_2)+
(\lambda_1-\lambda_2) r_{2}(\lambda_1)= \\ &
P_{r_2(\lambda_2)}-\lambda_2 r_2(\lambda_2)+(\lambda_2-\lambda_1)(r_{2}(\lambda_2)-r_{2}(\lambda_1))\stackrel{(11)}{<}   P_{r_2(\lambda_2)}-\lambda_2 r_2(\lambda_2)
\end{align*} 
contradicting~(\ref{hyp1}). So $r_2(\lambda_1)\leq r_2(\lambda_2)$, and the inequality $P_{r_2(\lambda_1)}\leq P_{r_2(\lambda_2)}$ follows. 

\item If on the other hand $OPT_{r_2(\lambda_1)}=OPT_{r_2(\lambda_2)}=M_{*}(p)$ we get again $P_{r_2(\lambda_1)}\leq P_{r_2(\lambda_2)}$ by~(\ref{f1}). 

\item The remaining case is when one of $OPT_{r_2(\lambda_1)}, OPT_{r_2(\lambda_2)}$ is $M_{*}(p)$, the other is not. 

Again the case $OPT_{r_2(\lambda_1)}=r_{2}(\lambda_1)\leq M_{*}(p)$, $OPT_{r_2(\lambda_2)}= M_{*}(p)$ yields directly, via~(\ref{f1}), the conclusion  $P_{r_2(\lambda_1)}\leq P_{r_2(\lambda_2)}$. So assume that $OPT_{r_2(\lambda_1)}= M_{*}(p)$ and 
$OPT_{r_2(\lambda_2)}=r_2(\lambda_2)\leq M_{*}(p)$. So $r_{2}(\lambda_2)\leq M_{*}(p)\leq r_{2}(\lambda_1)$. 

Because of this relation and the definition of function $r_2$ we have 
\begin{equation} 
P_{r_2(\lambda_2)}-\lambda_2 OPT_{r_{2}(\lambda_2)}\leq P_{r_2(\lambda_1)}-\lambda_2 OPT_{r_{2}(\lambda_1)}. 
\label{hyp2}
\end{equation} 
that is 
\begin{equation} 
P_{r_2(\lambda_2)}-\lambda_2 r_{2}(\lambda_2)\leq P_{r_2(\lambda_1)}-\lambda_2 M_{*}(p). 
\label{hyp3}
\end{equation}

But then 
\begin{align*}
& P_{r_2(\lambda_1)}-\lambda_2 OPT_{r_2(\lambda_1)}= 
P_{r_2(\lambda_1)}-\lambda_1 M_{*}(p)+(\lambda_1-\lambda_2)M_{*}(p)\stackrel{(10)}{\leq} P_{r_2(\lambda_2)}-\lambda_1 r_2(\lambda_2)+ \\ 
& (\lambda_1-\lambda_2)M_{*}(p)\leq  P_{r_2(\lambda_2)}-\lambda_2 r_2(\lambda_2)+ (\lambda_2-\lambda_1)r_2(\lambda_2)+
(\lambda_1-\lambda_2) M_{*}(p)=
P_{r_2(\lambda_2)}- \\ 
& \lambda_2 r_2(\lambda_2)+(\lambda_2-\lambda_1)(r_{2}(\lambda_2)-M_{*}(p))\stackrel{(11)}{<}   P_{r_2(\lambda_2)}-\lambda_2 r_2(\lambda_2)
\end{align*} 
contradicting~(\ref{hyp3}). So the proof of the first inequality is complete. 

\end{itemize} 
\textbf{Second inequality.} 

As for the second inequality, since $\lambda_1<\lambda_2\leq 1$ and $c_{OPT}(p)\geq 1$, we have
\begin{equation} 
\lambda_2-\lambda_1\leq \frac{(\lambda_2-\lambda_1)}{\lambda_1 \lambda_2}c_{OPT}(p)
\end{equation}
But this is equivalent to $\lambda_2+\frac{1}{\lambda_2}c_{OPT}(p) \leq \lambda_1+\frac{1}{\lambda_1}c_{OPT}(p)$, so 
\begin{equation} 
\lambda_2-1+\frac{1}{\lambda_2}c_{OPT}(p) \leq \lambda_1-1+\frac{1}{\lambda_1}c_{OPT}(p)
\end{equation} 
This shows that 
$r_3(\lambda_{2})$ is among the candidates for $r_{3}(\lambda_1)$. But this means that $\frac{P_{r_3(\lambda_1)}}{r_{3}(\lambda_1)}\leq  \frac{P_{r_3(\lambda_2)}}{r_{3}(\lambda_2)}$. 

\end{proof} 

\subsection{Proof of Lemma~\ref{lemma-opt2} (b).} 

\begin{proof} 
Choose $\lambda$ small enough such that $r_{0} \geq  (1-\lambda) (r_{0}-1)+ \lambda r_1 $, that is $\lambda(r_1-r_{0}+1)\leq 1$. Then $\lceil (1-\lambda) (r_{0}-1)+ \lambda r_1\rceil = r_{0}$. 

To show that for small values of $\lambda$ we have $r_2=r_{0}$ it is enough to show that for such values of $\lambda$ we have $P_{r_{0}}-\lambda OPT_{r_{0}}\leq P_{r}-\lambda OPT_r$ for all $r\geq r_{0}$. We proved above that $r_0\leq M_{*}(p)$ so, by Lemma~\ref{lemma-opt}, this is equivalent to $P_{r_{0}}-\lambda r_{0}\leq P_{r}-\lambda OPT_r$ for all $r\geq r_{0}$
This amounts to $\lambda(OPT_r-r_{0})\leq P_{r}-P_{r_{0}}$. Since $OPT_r=min(r, M_{*}(p))\geq r_0$, the term on the left-hand side is non-negative. If it is zero then the inequality is true for all values of $\lambda$, since $P_r\geq P_{r_0}=M_{*}(p)$. If, on the other hand $r_0<OPT_r$ then to make the inequality true, given that $OPT_r\leq M_{*}(p)$,  is enough that $\lambda<\frac{1}{M_{*}(p)-r_{0}}$. We have used here the fact that for $r>r_{0}$ we have $P_{r}-P_{r_{0}}\geq 1$. This is because $r_0$ is the \textbf{largest} value for which $P_{r_{0}}=M_{*}(p)$. Hence for $r>r_{0}$ we have $P_{r}\geq M_{*}(p)+1=P_{r_0}+1$. 

Now choosing $\lambda<\frac{1}{M_{*}(p)-r_{0}+1}$ is enough to satisfy both conditions.  
\end{proof}

\end{document}